\documentclass[10pt,journal,english,twocolumn]{IEEEtran}
\usepackage{blindtext}

\usepackage{babel}
\usepackage[babel]{microtype}
\usepackage[T1]{fontenc}

\usepackage{graphicx}
\usepackage{xcolor}
\usepackage{tikz}
\usepackage{pgfplots}

\usepackage{tabularx}
\usepackage{booktabs}
\usepackage{multirow}
\usepackage[keeplastbox]{flushend}
\usepackage[caption=false,font=footnotesize]{subfig}

\usepackage{amsmath}
\usepackage{amsfonts}
\usepackage{amssymb}
\usepackage{amsthm}
\usepackage{dsfont}
\usepackage{bm}
\usepackage{siunitx}

\usepackage[babel]{csquotes}
\usepackage[backend=biber,url=false,style=ieee,isbn=false,doi=true,eprint=true,dashed=false]{biblatex}
\bibliography{literature.bib}
\renewcommand*{\bibfont}{\footnotesize}

\usepackage{hyperref}
\hypersetup{
	pdfauthor={Karl-Ludwig Besser},
	pdftitle={On Fading Channel Dependency Structures with a Positive Zero-Outage Capacity},
	colorlinks=true,
	linkcolor=black,
	citecolor=black,
	allcolors=.,
	urlcolor=blue,
}

\usepackage[acronyms,nomain,xindy]{glossaries}
\makeglossaries
\newacronym{5g}{5G}{Fifth-Generation Wireless Systems}
\newacronym{ask}{ASK}{amplitude-shift keying}
\newacronym[plural={AWGN}, \glsshortpluralkey={AWGN}]{awgn}{AWGN}{additive white Gaussian noise}

\newacronym{ber}{BER}{bit error rate}
\newacronym{bler}{BLER}{block error rate}
\newacronym{bpsk}{BPSK}{binary phase-shift keying}

\newacronym{cdf}{CDF}{cumulative distribution function}
\newacronym{cnn}{CNN}{convolutional neural network}
\newacronym{csi}{CSI}{channel-state information}
\newacronym{csit}{CSI-T}{channel-state information at the transmitter}

\newacronym{ecc}{ECC}{error-correcting code}

\newacronym{ff}{FF}{feed-forward}
\newacronym{ffnn}{FF-NN}{feed-forward neural network}
\newacronym{fsk}{FSK}{frequency-shift keying}

\newacronym{gm}{GM}{Gaussian mixture}

\newacronym{ic}{IC}{interference channel}
\newacronym{iid}{i.\,i.\,d.}{independent and identically distributed}
\newacronym{il}{IL}{information leakage}
\newacronym{iot}{IoT}{internet of things}

\newacronym{lb}{LB}{lower bound}

\newacronym{mac}{MAC}{multiple access channel}
\newacronym{map}{MAP}{maximum a posteriori}
\newacronym{mimo}{MIMO}{multiple-input multiple-output}
\newacronym{ml}{ML}{machine learning}
\newacronym{mop}{MOP}{multi-objective programming problem}
\newacronym{mrc}{MRC}{maximum ratio combining}
\newacronym{mse}{MSE}{mean squared error}

\newacronym{nn}{NN}{neural network}
\newacronym{nve}{NVE}{normalized validation error}

\newacronym{pdf}{PDF}{probability density function}
\newacronym{pwc}{PWC}{polar wiretap code}

\newacronym{qam}{QAM}{quadrature amplitude modulation}

\newacronym{rbf}{RBF}{radial basis function}
\newacronym{relu}{ReLU}{rectified linear unit}
\newacronym{ris}{RIS}{reconfigurable intelligent surface}
\newacronym{rnn}{RNN}{recurrent neural network}

\newacronym{sc}{SC}{selection combining}
\newacronym{sgd}{SGD}{stochastic gradient descent}
\newacronym{sgp}{SGP}{secure goodput}
\newacronym{sinr}{SINR}{signal-to-interference-plus-noise ratio}
\newacronym{siso}{SISO}{single-input single-output}
\newacronym{snr}{SNR}{signal-to-noise ratio}
\newacronym{svm}{SVM}{support vector machine}

\newacronym{tin}{TIN}{treating interference as noise}
\newacronym{tvd}{TVD}{total variation distance}

\newacronym{ub}{UB}{upper bound}
\newacronym{urllc}{URLLC}{ultra-reliable low latency communication}

\newacronym{zoc}{ZOC}{zero-outage capacity}
\setacronymstyle{long-short}

\pgfplotsset{compat=newest}
\pgfplotsset{plot coordinates/math parser=false}

\usetikzlibrary{patterns,arrows,positioning,external}
\usepgfplotslibrary{patchplots}

\newcommand{\nice}{SYM{}}

\let\vec\bm
\newcommand{\diff}{\ensuremath{\mathrm{d}}}
\newcommand{\expect}[2][]{\ensuremath{\mathbb{E}_{#1}\left[#2\right]}}
\newcommand{\inv}[1]{\ensuremath{#1^{-1}}}
\newcommand{\abs}[1]{\ensuremath{\left|#1\right|}}
\newcommand{\positive}[1]{\ensuremath{\left[#1\right]^{+}}}

\newcommand{\reals}{\ensuremath{\mathds{R}}}
\newcommand{\complexes}{\ensuremath{\mathds{C}}}
\newcommand{\e}{\ensuremath{\mathrm{e}}}

\newcommand{\X}{\ensuremath{{X}}}

\newcommand{\Y}{\ensuremath{{Y}}}

\newcommand{\Z}{\ensuremath{{Z}}}

\newcommand{\lx}{\ensuremath{{\lambda_{1}}}}
\newcommand{\ly}{\ensuremath{{\lambda_{2}}}}

\newcommand{\snr}{\ensuremath{\rho}}
\newcommand{\snrx}{\ensuremath{\rho_{1}}}
\newcommand{\snry}{\ensuremath{\rho_{2}}}

\newcommand{\zerorate}{\ensuremath{R^{0}}}
\newcommand{\epsrate}{\ensuremath{R^{\varepsilon}}}

\DeclareMathOperator*{\argmin}{arg\,min}
\DeclareMathOperator{\median}{median}
\DeclareMathOperator{\mode}{mode}

\theoremstyle{plain}%
\newtheorem{thm}{Theorem}%
\newtheorem{lem}[thm]{Lemma}

\newtheorem{cor}[thm]{Corollary}

\theoremstyle{definition}
\newtheorem{defn}{Definition}%

\theoremstyle{remark}
\newtheorem{rem}{Remark}
\newtheorem*{rem*}{Remark}

\definecolor{plot0}{HTML}{1bb34c}
\definecolor{plot1}{HTML}{242bb3}
\definecolor{plot2}{HTML}{b3120e}
\definecolor{plot3}{HTML}{cc9f2b}
\definecolor{plot4}{HTML}{27AFB3}

\definecolor{coolblack}{rgb}{0.0, 0.18, 0.39}

\definecolor{change}{HTML}{0096b8}

\title{On Fading Channel Dependency Structures with a Positive Zero-Outage Capacity}
\author{Karl-Ludwig Besser, \IEEEmembership{Student Member, IEEE}, Pin-Hsun Lin, \IEEEmembership{Member, IEEE}, and Eduard A. Jorswieck, \IEEEmembership{Fellow, IEEE}
	
\thanks{Parts of this work have been presented at the 54th Asilomar Conference on Signals, Systems, and Computers, November 2020~\cite{Besser2020asilomar}.}
\thanks{The authors are with the Institute of Communications Technology, Technische Universit\"at Braunschweig, 38106 Braunschweig, Germany (email: \{{k.besser}, {p.lin}, {e.jorswieck}\}@tu-bs.de).}
\thanks{The work of K.-L. Besser is supported in part by the German Research Foundation (DFG) under grant JO\,801/23-1. The work of E. Jorswieck is supported in part by DFG under grant JO\,801/25-1.}
}

\begin{document}
\maketitle

\begin{abstract}%
	With emerging wireless technologies like 6G, many new applications like autonomous systems evolve which have strict demands on the reliability and latency of data communications.
	In the scenario of the commonly investigated independent slow fading links, the zero-outage capacity (ZOC) is zero and retransmissions are therefore inevitable.
	In this work, we show that a positive ZOC can be achieved under the same setting of slow fading with constant transmit power and without perfect channel state information at the transmitter, if the joint distribution of the channel gains follows certain structures. This allows reliable reception without any outages, thus not requiring retransmissions.
	Based on a systematic copula approach, we show that there exists a set of dependency structures for which positive ZOCs can be achieved for both maximum ratio combining (MRC) and selection combining (SC). 
	We characterize the maximum ZOC within a finite number of bits.
	The results are evaluated explicitly for the special cases of Rayleigh fading and Nakagami-$m$ fading in order to quantify the ZOCs for common fading models.
\end{abstract}
\begin{IEEEkeywords}
Copula, Joint distributions, Fading channels, Delay-limited capacity, Outage probability.
\end{IEEEkeywords}

\section{Introduction}
Current visions for emerging technologies like 6G include new applications like autonomous systems which have extremely high reliability constraints~\cite{Saad2019,Park2020}. It is therefore of great interest, how these strict reliability constraints can be achieved.

For slow-fading channels, a common performance metric is the $\varepsilon$-outage capacity.
It is defined as the maximum transmission rate for which the outage probability is not greater than $\varepsilon$.
Of particular interest is the \gls{zoc}, i.e., the maximum transmission rate at which we can transmit \emph{without any} outages.
This will be especially relevant in the context of \gls{urllc}~\cite{Park2020,Bennis2018,Ding2021urllc}.

It is well-known that the \gls{zoc} is zero in most of the common scenarios considered in the literature, e.g., single link \gls{siso} with independent Rayleigh fading~\cite{Tse2005}.
In \cite{Zhu2019}, it is shown that a negative correlation improves the outage probability compared to independent channels.
The authors consider the special scenario of a dual-branch system with linearly correlated log-normal fading channels.
A similar observation is made in \cite{Zhu2018}.
Furthermore, it was shown in \cite{Jorswieck2019}, that the \gls{zoc} can be positive for two Rayleigh fading links without perfect \gls{csit}, if the two channels are negatively correlated.
Taking more general dependency structures into account, bounds on the $\varepsilon$-outage capacity for communication systems with $n$ dependent fading links are derived in \cite{Besser2020twc}. It is shown that the upper-bound on the $\varepsilon$-outage capacity, i.e., for the best-case joint distribution, is strictly positive.
This implies that there exists at least one joint distribution for which the \gls{zoc} is positive for given arbitrary marginals. However, it remains an open question, if there exist multiple joint distributions that achieve a positive \gls{zoc}.
Additionally, it is unclear, if there exist joint distributions for which the \gls{zoc} is positive but strictly less than the upper bound derived in \cite{Besser2020twc}.
A positive \gls{zoc} implies that a transmission without any outages is possible. This is relevant, e.g., in the context of \gls{urllc}, since no retransmissions are necessary which allows a very low latency. Besides, if there exists only a single joint distribution with a positive \gls{zoc}, this would be an unstable operating point.
We are therefore also interested in other joint distributions achieving positive \glspl{zoc}.

In this work, we will answer both of these open questions: 
We show that there exists an infinite number of joint distributions, for which the \gls{zoc} is strictly between zero and the best-case upper bound.
The proof is constructive in the way that we explicitly give a parametrized family of joint distributions that achieve positive \glspl{zoc}.
Our results hold for general fading distributions, and are evaluated for the special case of Rayleigh fading and Nakagami-$m$ fading.

From the results of this work and previous work, e.g., \cite{Besser2020twc}, it becomes apparent that the outage performance of communication systems can be significantly improved by dependent channels.
In view of the strict reliability requirements that current and future 6G applications have~\cite{Tataria2021,Ding2021urllc}, this indicates that active dependency control is a promising direction of research to enable ultra-reliable communications.
Emerging technologies like \glspl{ris}~\cite{DiRenzo2019} could allow for such an active dependency control.

This work is concerned with the theoretical limits of the outage performance for dependent fading channels, which will serve as an upper bound for practical and heuristic future communication system designs.
The basis of our derivations is copula theory~\cite{Nelsen2006}, which allows a flexible modeling of dependency structures between random variables.
Copulas have been used in communications before to model dependency between channels.
One of the first works on applying this technique in the context of wireless communications is \cite{Ritcey2007}.
There, a Clayton copula is used to introduce a new fading model for Nakagami-$m$ fading with tail dependence.
In \cite{Kitchen2010,Gholizadeh2015}, copulas have been used to model dependency in \gls{mimo} channels.
In \cite{Peters2014}, it was shown on measurements that real channels can show a tail dependency which can be described by copulas.
The outage probability of the two-user Rayleigh fading \gls{mac}, where the two links follow a specific copula, is evaluated in \cite{Ghadi2020}.
Copulas have also been used to model interference in \gls{iot} networks~\cite{Zheng2019} and in the area of physical layer security~\cite{Besser2020tcom,Besser2020wsa,Ghadi2021}.
General bounds on the outage performance for dependent slow-fading channels can be found in \cite{Besser2020part2}.

The tools from copula theory are not only helpful for analyzing the outage capacity for slow fading channels, but also for the ergodic capacities for fast fading multi-user channels~\cite{Besser2020part1}.
In particular, when the multi-user channel has the same marginal property, copulas have been used to derive capacities regions for Gaussian interference channels, Gaussian broadcast channels and the secrecy capacity for Gaussian wiretap channels \cite{Lin2019icc,Lin2019isit,Lin2020tcom}.

However, in this work, we focus on the \gls{zoc} for $n$ \emph{dependent} slow-fading links. Our contributions are summarized as follows.
\begin{itemize}
	\item We show that the \gls{zoc} in a multi-connectivity setting can be strictly positive for dependent channels with arbitrary marginal fading distributions, even when only statistical \gls{csit} is available and a constant transmit power is used.
	We consider both \gls{mrc} and \gls{sc} at the receiver.
	\item We show that the \gls{zoc} is maximized by countermonotonic channel gains in the case of two dimensions. For the general case, we show that there exists an infinite number of joint distributions, for which the \gls{zoc} is strictly between zero and the best-case (upper) bound. The proof is constructive in the way that we explicitly state a parameterized family of joint distributions that achieve positive \glspl{zoc}.
	\item We provide an explicit expression for the maximum \gls{zoc} for \gls{sc} at the receiver with $n$ homogeneous links.
	For $n$ heterogeneous links, we give an implicit characterization of the maximum \gls{zoc}.
	\item We derive inner and outer bounds for the maximum \gls{zoc} in the general case of $n>2$ channels with homogeneous channel gains $\X$ and \gls{mrc} at the receiver. The bounds are within a finite gap for all $n$ equal to $\log_2\left({\expect{\X}}\right) - \log_2\left({\inv{F_{\X}}(1/\e)}\right)$.
	\item All results are evaluated for the special cases of Rayleigh fading and Nakagami-$m$ fading. In addition, we provide all plots presented in this work as interactive versions in \cite{BesserGitlab}.
\end{itemize}

The practical implications of the results are the following.
If one is able to tune the dependency structure between different channels, e.g., by using smart relays~\cite{Wang2008} or reconfigurable meta-surfaces~\cite{DiRenzo2019}, it is possible to transmit data over fading channels without any outages.
The second result implies a certain robustness.
Assume that we are able to parameterize the dependency structure within a desired range.
Even if the designed parameters are not set perfectly, the \gls{zoc} can still be positive, if one hits one of the infinitely many joint distributions with positive \gls{zoc}.

The rest of the paper is organized as follows.
In Section~\ref{sec:background-system-model}, we state the system model and problem formulation of this work.
We also introduce some necessary mathematical background from copula theory.
Some general observation and results are given in Section~\ref{sec:general}.
The two diversity combining techniques \gls{mrc} and \gls{sc} are investigated in detail in Sections~\ref{sec:mrc} and \ref{sec:sc}, respectively.
Finally, Section~\ref{sec:conclusion} concludes the paper.

\textit{Notation:}
Throughout this work, we use the following notation. 
Random variables are denoted in capital letters, e.g., $\X$, and their realizations in small letters, e.g., $x$.
Vectors are written in boldface letters, e.g., $\vec{X}$.
We use $F$ and $f$ for a probability distribution and its density, respectively. The expectation is denoted by $\mathbb{E}$ and the probability of an event by $\Pr$. It is assumed that all considered distributions are continuous. The uniform distribution on the interval $[a,b]$ is denoted as $\mathcal{U}[a,b]$. The normal distribution with mean $\mu$ and variance $\sigma^2$ is denoted as $\mathcal{N}(\mu, \sigma^2)$.
The derivative of a function $f$ is denoted by $f'$.
As a shorthand, we use $\positive{x}=\max\left[x, 0\right]$.
The real numbers, non-negative real numbers, and extended real numbers are denoted by $\reals$, $\reals_{+}$, and $\bar{\reals}$, respectively. Logarithms, if not stated otherwise, are assumed to be with respect to the natural base.
\section{Preliminaries and System Model}\label{sec:background-system-model}
We consider a slow fading channel with $n$ receive antennas.
The symbol ${M}\in\complexes$ is transmitted with a rate $R$ over the slow fading links ${H}_i\in\complexes$, $i=1,\dots{},n$.
The received signal $\vec{Y}=(\Y_1, \dots{}, \Y_n)\in\complexes^n$ is given as
\begin{equation}
\Y_i = {H}_i {M} + {N}_i\,,
\end{equation}
where ${N}_i$ are independent complex Gaussian noise terms with zero mean and variance $\sigma_i^2$. The transmit \glspl{snr} of the individual channels are given as $\snr_i=P/\sigma_i^2$, where $P$ is the average transmit power constraint.

We assume that the receiver has perfect \gls{csi} while the transmitter only has statistical \gls{csi}. The receiver applies the diversity combining strategy $L: \reals_{+}^{n}\to\reals$. In this work, we consider strategies in the form $L(\X_1, \dots{}, \X_n)$ with $\X_i = \snr_i\abs{{H}_i}^2$.
In this case, an outage occurs, if the instantaneous channel capacity is less than the rate $R$ used for the transmission, i.e., $\log_2\left(1+ L(\X_1, \dots{}, \X_n)\right)<R$.
The maximum rate $\epsrate$ for which the probability of such an outage is less than $\varepsilon$ is called \emph{$\varepsilon$-outage capacity} and defined as~\cite{Tse2005}
\begin{equation}\label{eq:def-eps-rate}
\epsrate = \sup_{R\geq 0}\left\{R \;|\; \Pr\Big(\log_2(1+L(\X_1, \dots{}, \X_n)) < R\Big) \leq \varepsilon\right\}.
\end{equation}

This can be reformulated as the (receive) \gls{snr} optimization
\begin{equation}\label{eq:def-prob-eps-outage-capac}
s^\star(\varepsilon) = \sup_{s\geq 0} \left\{s \;|\; \Pr\left(L(\X_1, \dots{}, \X_n)<s\right) \leq \varepsilon\right\}
\end{equation}
with $\epsrate = \log_2(1 + s^\star(\varepsilon))$. %

\subsection{Problem Formulation}\label{sub:problem-formulation}
We consider a communication scenario with $n$ slow-fading channels which can be dependent. The receiver has perfect \gls{csi}, while we only assume statistical \gls{csit}.
The quantity of interest is the zero-outage capacity $\zerorate$, which is given as the maximum rate fulfilling
\begin{equation}\label{eq:general-cond-zero-prob}
\Pr\left(L(\X_1, \dots{}, \X_n)<2^{\zerorate}-1\right) = 0\,.
\end{equation}
The main question that we will answer in this work is:
\emph{Given marginal distributions of the individual wireless fading channels, do there exist joint distributions (and how many) which achieve positive \glspl{zoc} and what \glspl{zoc} can be achieved by different diversity combining techniques?}
As a consequence, this also includes the maximum \gls{zoc} with respect to all joint distributions with the given marginals.

By the construction in \cite{Besser2020twc}, we know that there exists at least one joint distribution with positive zero-outage capacity. However, a singleton is an unstable operating point and it is therefore of interest, how robust a positive \gls{zoc} is with respect to the joint distribution. In addition, we are interested in the maximum \gls{zoc} since it shows what the best case performance can be.

\subsection{Mathematical Background}
To describe and analyze the structure of joint distributions, we will use tools from copula theory~\cite{Nelsen2006}, which we introduce in the following.

\begin{defn}[Copula]\label{def:copula}
	A copula is an $n$-dimensional distribution function with standard uniform marginals.
\end{defn}

The practical relevance of copulas stems from Sklar's theorem, which we restate in the following Theorem~\ref{thm:sklar}.
\begin{thm}[{Sklar's Theorem~\cite[Thm.~2.10.9]{Nelsen2006}}]\label{thm:sklar}
	Let $H$ be an $n$-dimensional distribution function with margins $F_{1}, \dots{}, F_{n}$. Then there exists a copula $C$ such that for all $x\in\bar{\reals}^n$,
	\begin{equation}\label{eq:sklar-joint-dist}
	H(x_1, \dots{}, x_n) = C(F_{1}(x_1), \dots{}, F_{n}(x_n))\,.
	\end{equation}
	If $F_{1}, \dots{}, F_{n}$ are all continuous, then $C$ is unique. Conversely, if $C$ is a copula and $F_{1}, \dots{}, F_{n}$ are distribution functions, then $H$ defined by \eqref{eq:sklar-joint-dist} is an $n$-dimensional distribution function with margins $F_{1}, \dots{}, F_{n}$.
\end{thm}
This theorem implies that copulas can be used to describe dependency structures between random variables, regardless of their marginal distributions. This allows us to separate the dependency structure (described by the copula $C$) from the marginal distributions $F_{1}, \dots{}, F_{n}$.
We will see in the following section that the \gls{zoc} depends on the underlying copula between the channel gains.
{Some} results in this work are based on the Fr\'{e}chet-Hoeffding bounds for copulas, which we state in the following theorem.
\begin{thm}[{Fr\'{e}chet-Hoeffding Bounds~\cite[Thm.~2.10.12]{Nelsen2006}}]\label{thm:frechet-hoeffding-bounds}
	Let $C$ be a copula. Then for every $u\in[0, 1]^n$
	\begin{equation}
	W(u) \leq C(u) \leq M(u)
	\end{equation}
	with %
	\begin{align}
	W(u) &= \max\left\{u_1 + \cdots{} + u_n - n + 1, 0\right\}\,,\\
	M(u) &= \min\left\{u_1, \dots{}, u_n\right\}\,.
	\end{align}
\end{thm}
In the case that $n=2$, $W$ is a copula and two random variables whose joint distribution follows the copula $W$ are called countermonotonic.
The upper bound $M$ is a copula for all $n$ and random variables that follow $M$ are called comonotonic~\cite{Nelsen2006}.
\section{General Considerations and Results}\label{sec:general}

In this section, we will make some general observations and derivations about the considered problem. Throughout this work, we will assume that the diversity combining function $L$ only depends on the channel gains $\X_i=\snr_i \abs{{H}_i}^2$ and is non-decreasing in each variable.

The underlying observation for our derivations in this work is the following. As stated in Section~\ref{sec:background-system-model}, we know that the outage probability corresponds to the probability of the event $L(\X_1, \dots{}, \X_n)<s$. An equivalent way of expressing this is via the integral of the joint distribution $F_{\X_1,\dots{},\X_n}$ over the area
\begin{equation}
	\mathcal{S}=\left\{(x_1, \dots{}, x_n)\in\reals_{+}^{n}\;|\;L(x_1, \dots{}, x_n)<s\right\}\,.
\end{equation}
Then the outage probability, given $R$, can be written as
\begin{equation}\label{eq:outage-prob-integral}
	\varepsilon = \int_{\mathcal{S}}\!\diff{C(F_{\X_1}(x_1), \dots{}, F_{\X_n}(x_n))}\,,
\end{equation}
where we use the copula representation of the joint distribution based on Theorem~\ref{thm:sklar}.

Recall that our goal is to have zero outages, i.e., $\varepsilon=0$.
With reference to \eqref{eq:outage-prob-integral}, this is achieved if the probability of the joint distribution is zero in $\mathcal{S}$.
The joint \gls{cdf} is zero if and only if its copula $C$ is zero, and the corresponding area of $(\X_1, \dots{}, \X_n)$ can be written as
\begin{equation}\label{eq:def-area-b}
	\mathcal{B}=\left\{(x_1, \dots{}, x_n) \;|\; C(F_{\X_1}(x_1), \dots{}, F_{\X_n}(x_n))=0\right\}\,.
\end{equation}
In other words, $\mathcal{S}$ is not inside the support of $(\X_1, \dots{}, \X_n)$. %
This idea is exemplarily shown in Fig.~\ref{fig:areas-max-zoc} for the two diversity schemes \gls{mrc} and \gls{sc}.
Detailed explanations and results for both diversity schemes will be given in the following sections.

\begin{figure}
	\centering
	\begin{tikzpicture}
	\begin{axis}[
		width=.95\linewidth,
		height=.24\textheight,
		xlabel={$x_1$},
		ylabel={$x_2$},
		domain=0:3,
		xmin=0.05,
		xmax=3,
		ymin=0,
		axis on top,
		xtick={0, 1.534, 2.917},
		xticklabels={0, $s^\star_{\text{SC}}$, $s^\star_{\text{MRC}}$},
		ytick={0, 1.534, 2.917},
		yticklabels={0, $s^\star_{\text{SC}}$, $s^\star_{\text{MRC}}$},
		]
		
		\addplot[plot4, thick, fill=plot4, fill opacity=.2, area legend] table [x=x, y=boundary] {data/boundary-example.dat}\closedcycle;
		\addlegendentry{$\mathcal{B}$};
		\node[text=plot4] at (axis cs: 2.8,.55) {$\mathcal{B}$};

		\addplot[plot2, thick, area legend, fill=plot2, fill opacity=.2] {2.917-x}\closedcycle;
		\addlegendentry{$\mathcal{S}_{\text{MRC}}$};
		\node[text=plot2] at (axis cs: .2,2) {$\mathcal{S}_{\text{MRC}}$};
		\addplot[plot1, dashed, thick, area legend, fill=plot1, fill opacity=.2]  coordinates {
			(0, 0)
			(1.534, 0)
			(1.534, 1.534)
			(0, 1.534)
		}\closedcycle;
		\addlegendentry{$\mathcal{S}_{\text{SC}}$};
		\node[text=plot1] at (axis cs: .75,.75) {$\mathcal{S}_{\text{SC}}$};
	\end{axis}
\end{tikzpicture}
	\caption{Areas corresponding to the \gls{zoc}. Area $\mathcal{B}$ shows the area where $F_{\X_1,\X_2}=0$. Area $\mathcal{S}_{\text{MRC}}$ shows the integration area from \eqref{eq:outage-prob-integral} to calculate the outage probability when \gls{mrc} is used at the receiver. The value $s^\star_{\text{MRC}}$ denotes the maximum value of $s$ such that $\mathcal{S}_{\text{MRC}}$ is still a subset of $\mathcal{B}$. The analogue for \gls{sc} at the receiver is denoted by the index \enquote{SC}.}
	\label{fig:areas-max-zoc}
\end{figure}
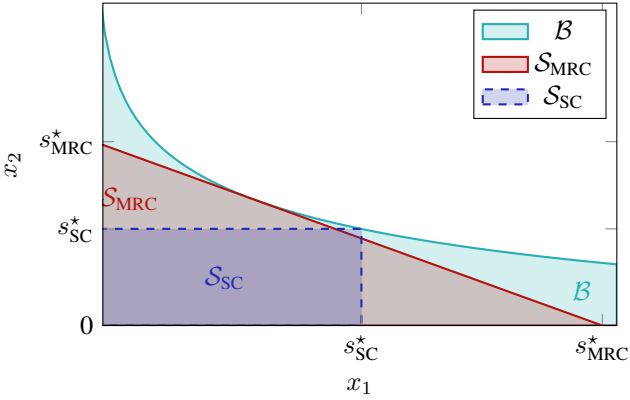

\subsection{Two-Dimensional Case}
First, we show that the \gls{zoc} $\zerorate$ is upper bounded by countermonotonic channels for $n=2$.
\begin{cor}[{Maximum \Gls{zoc}}]\label{cor:max-zoc-countermonotonic}
	The maximum \gls{zoc} for two links with channel gains $\X_1$ and $\X_2$ is attained by countermonotonic random variables, i.e., their joint distribution follows copula $W$.
\end{cor}
\begin{proof}
	The proof can be found in Appendix~\ref{app:proof-cor-max-zero-out}.
\end{proof}

\begin{rem}[{Minimum \Gls{zoc}}]
	It is well-known that the \gls{zoc} can be zero, e.g., for channels with independent links~\cite{Tse2005}. Since the capacity is a non-negative quantity, we can conclude that zero is the lower bound on the \gls{zoc}.
\end{rem}

\subsection{General Case}
The extension to the general $n$-dimensional case with $n>2$ is not straightforward, since $W$ is only a valid copula for $n=2$.
However, we can reformulate the problem in the general case using the discussions from above.
Our goal is to find the maximum value of $R$ in \eqref{eq:def-eps-rate} (or equivalently, $s^\star$ in \eqref{eq:def-prob-eps-outage-capac}), such that $\mathcal{S}$ is still a subset of $\mathcal{B}$.
Note that the boundary of $\mathcal{B}$ is not necessarily convex.
We, therefore, rewrite the optimization problem \eqref{eq:def-prob-eps-outage-capac} as
\begin{equation*}%
	s^\star = \max_{(x_1, \dots{}, x_n)\in \mathcal{B}} L(x_1, \dots{}, x_n)\,.
\end{equation*}
Since we know from the monotonicity of $L$ that the maximum will be on the boundary of $\mathcal{B}$ defined by $B$, $s^\star$ can also be written as minimizing the function $L$ over the boundary $B$ as
\begin{equation}\label{eq:opt-problem-boundary}
	\begin{split}
		s^\star = \min_{(x_1, \dots{}, x_n)} &L(x_1, \dots{}, x_n)\\
		\text{s.\,t.~} &B(x_1, \dots{}, x_n) = 0\,.
	\end{split}%
\end{equation}
The structure of the boundary function $B$ is determined by the underlying joint distribution and we will give particular examples in the following sections.
Throughout the rest of this work, we will assume that $B$ is a smooth function.
\section{Maximum Ratio Combining}\label{sec:mrc}

First, we will investigate the case that the receiver applies \gls{mrc} as the diversity combining technique. In this case, the combination function $L$ is the sum of the channel gains $\X_i$~\cite{Tse2005}%
\begin{equation*}
	L_{\text{MRC}}(\X_1, \dots{}, \X_n) = \sum_{i=1}^{n} \X_i\,.
\end{equation*}

In order to introduce the basic concepts, we will start with the two-dimensional case and then extend the results to the general $n$-dimensional case.
The results will be illustrated with two examples, namely Rayleigh fading and Nakagami-$m$ fading.

\subsection{Two-Dimensional Case}
We start with the two-dimensional case $n=2$.
From Corollary~\ref{cor:max-zoc-countermonotonic}, we know that the maximum \gls{zoc} is attained for countermonotonic $(\X_1, \X_2)$.
In the following theorem, we show that for each value $c$ between zero and the maximum \gls{zoc}, there exists a joint distribution for which the \gls{zoc} is equal to $c$, i.e., the \gls{zoc} is continuous with respect to the joint distribution.

\begin{thm}[Zero-Outage Capacities for Two Links with \Gls{mrc}]\label{thm:mrc-achievable-zero-out-two-links}
	Let $\X_1$ and $\X_2$ be non-negative continuous random variables representing the channel gains of two communication links. The receiver applies \gls{mrc} as the diversity combining technique.
	Then there exist joint distributions of $\X_1$ and $\X_2$ with the following zero-outage capacities for $t\in[0, 1]$
	\begin{equation}\label{eq:mrc-zero-outage-capacity-two-links}
		\zerorate(t) = \log_2\left(1 + \min\left\{\inv{F_{\X_1}}(t),\ \inv{F_{\X_2}}(t),\ x^\star+B_t(x^\star)\right\}\right)
	\end{equation}
	with
	\begin{equation}\label{eq:def-b}
		B_t(x) = \inv{F_{\X_2}}(t-F_{\X_1}(x))
	\end{equation}
	and %
	\begin{equation}\label{eq:condition-x-star}
		\begin{split}
		x^\star = \argmin_{x\geq 0}\big\{x+B_t(x) \;|\; &f_{\X_1}(x)=\\[-1.5ex] &\;f_{\X_2}\left(\inv{F_{\X_2}}(t-F_{\X_1}(x))\right)\big\}.
		\end{split}
	\end{equation}
\end{thm}
\begin{proof}
	The proof can be found in Appendix~\ref{app:proof-thm-achievable-zoc}.
\end{proof}
The parameter $t$ characterizes the joint distribution of the two fading links. For $t=0$ and $t=1$, $\X_1$ and $\X_2$ are comonotonic and countermonotonic, respectively. The maximum \gls{zoc} is therefore achieved for $t=1$.

\begin{rem}[Ambiguity of the Dependency Structure]\label{rem:ambiguity-copula-t}
	In the proof of Theorem~\ref{thm:mrc-achievable-zero-out-two-links}, we used the copula $C_t$ from \cite[Chap.~3.2.2]{Nelsen2006}.
	However, the only property we actually require for the proof is the fact that
	$C_t(a, b) = 0$ for $a+b\leq t$.
	Thus, any copula, that has this property, achieves the same result.
	Another example having this property is a generalization of the circular copula~\cite[Eq.~(3.1.5)]{Nelsen2006}
	\begin{equation*}
		C_{t, \text{circ}}(a, b) = \begin{cases}
			M(a, b) & \text{if } \abs{a-b} > t\\
			W(a, b) & \text{if } \abs{a+b-1} < 1-t\\
			\frac{a+b}{2} - \frac{t}{2} & \text{otherwise,}
		\end{cases}
	\end{equation*}
	which is illustrated in Fig.~\ref{fig:copula-circular} in Appendix~\ref{app:proof-thm-achievable-zoc}.
	There are also families of absolutely continuous copulas whose support is not the full unit square. An example is the Clayton copula given by~\cite[Eq.~(4.2.1)]{Nelsen2006}
	\begin{equation*}
		C_{\theta, \text{clay}}(a, b) = \left(\max\left[a^{-\theta} + b^{-\theta} - 1, 0\right]\right)^{-\frac{1}{\theta}},
	\end{equation*}
	with $\theta\in[-1, \infty)\setminus\{0\}$ for negative values of the parameter $\theta$.
	
	In Fig.~\ref{fig:clayton-rayleigh-example}, a numerical example to illustrate the structure of the joint distribution is shown where we show \num{2000} random samples of two Rayleigh fading channel gains with \glspl{snr} $\snr_1=\snr_2=\SI{0}{\decibel}$.
	Their joint distribution in Fig.~\ref{fig:clayton-rayleigh-t-0.75} follows the Clayton copula with $\theta=-0.75$.
	It can easily be seen that there are no realizations for which both channel gains are close to zero simultaneously.
	This enables a positive \gls{zoc}, cf. Fig.~\ref{fig:areas-max-zoc}.
	In contrast, the \gls{zoc} is zero for independent channels as can be seen in Fig.~\ref{fig:clayton-rayleigh-t0}.
	\begin{figure}[!htb]
		\subfloat[{Independence Copula}]{\begin{tikzpicture}%
	\begin{axis}[
		width=.5\linewidth, %
		height=.22\textheight, %
		xlabel={Channel Gain $\X_1$},
		ylabel={Channel Gain $\X_2$},
		xmin=0,
		xmax=8,
		ymin=0,
		ymax=6,
		ymajorgrids,
		xmajorgrids,
		xminorgrids,
		grid style={line width=.1pt, draw=gray!20},
		major grid style={line width=.25pt,draw=gray!40},
		]
		\addplot[plot0,thick,mark=*, only marks] table[x=x,y=y] {data/samples_rayleigh_clayton_theta0.00_snr0.00.dat};
	\end{axis}
\end{tikzpicture}
			\label{fig:clayton-rayleigh-t0}
		}
		\hfill
		\subfloat[{Clayton Copula with $\theta=-0.75$}]{\begin{tikzpicture}%
	\begin{axis}[
		width=.5\linewidth, %
		height=.22\textheight, %
		xlabel={Channel Gain $\X_1$},
		ylabel={Channel Gain $\X_2$},
		xmin=0,
		xmax=8,
		ymin=0,
		ymax=6,
		ymajorgrids,
		xmajorgrids,
		xminorgrids,
		grid style={line width=.1pt, draw=gray!20},
		major grid style={line width=.25pt,draw=gray!40},
		]
		\addplot[plot0,thick,mark=*, only marks] table[x=x,y=y] {data/samples_rayleigh_clayton_theta-0.75_snr0.00.dat};
	\end{axis}
\end{tikzpicture}
			\label{fig:clayton-rayleigh-t-0.75}
		}
		\caption{Random joint Rayleigh fading realizations with different copulas underlying the joint distribution. The \gls{snr} of both channel gains is set to $\snr=\SI{0}{\decibel}$.}\label{fig:clayton-rayleigh-example}
	\end{figure}
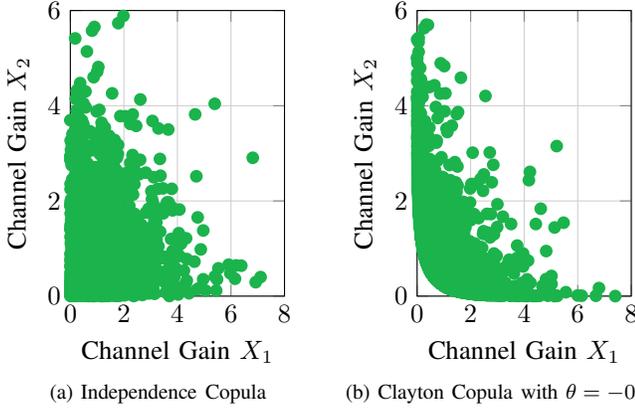
\end{rem}

\subsection{General Case}\label{sub:mrc-general-n-links}
From Theorem~\ref{thm:mrc-achievable-zero-out-two-links}, it can be seen that there exists an infinite number of joint fading distributions that achieve a positive \gls{zoc}.
This answers the first part of our question from Section~\ref{sub:problem-formulation} and shows that the set of joint distributions with a positive \gls{zoc} is not a singleton.
In the following, we will therefore only focus on the \emph{maximum} \gls{zoc} with respect to the set of joint distributions.

As shown in Corollary~\ref{cor:max-zoc-countermonotonic}, for the two-dimensional case, the maximum is achieved for countermonotonic channel gains, i.e., $\X_1$ and $\X_2$ follow the $W$-copula, cf.~Theorem~\ref{thm:frechet-hoeffding-bounds}.
Unfortunately, the extension to the general $n$-dimensional case is not straightforward since the Fr\'{e}chet-Hoeffding lower bound $W$ is not a copula anymore for $n>2$.
The most general case of arbitrary marginal distributions $F_{\X_1}, \dots{}, F_{\X_n}$ with $n>2$, therefore, remains an open problem.
However, we will derive new results for the $n$-dimensional case under the following assumptions.
We only consider the homogeneous case, i.e., all marginal distributions are the same, $F_{\X_1}=\cdots{}=F_{\X_n}=F_{\X}$. In addition, the distribution function $F_{\X}$ fulfills the following definition.
\begin{defn}[$B$-\nice{} Distribution]\label{def:nice-distribution}
	Let $\X_1, \dots{}, \X_n$ be non-negative continuous random variables with distribution function $F_{\X_1}=\dots{}=F_{\X_n}=F_{\X}$ with finite first moment. The distribution function $F_{\X}$ is a $B$-\nice{} distribution, if the solution to the optimization problem
	\begin{equation}\label{eq:opt-problem-nice-dist}
		\begin{split}
			s^\star = \min_{(x_1, \dots{}, x_n)} &\sum_{i=1}^{n} x_i\\
			\text{s.\,t.~} &B(x_1, \dots{}, x_n) = 0
		\end{split}
	\end{equation}
	lies on the identity line, i.e., $s^\star = n x^\star$ with $B(x^\star, \dots{}, x^\star)=0$.
\end{defn}

A characterization of some $B$-\nice{} distribution functions for two relevant $B$ is given in the following lemmas.
The two considered $B$ in Lemmas~\ref{lem:dist-cond-w} and \ref{lem:dist-cond-arch-copula} arise from the Fr\'{e}chet-Hoeffding lower bound and the lower bound on Archimedean copulas~\cite[Prop.~4.6]{McNeil2009}, respectively.
{They will later be needed in Theorems~\ref{thm:mrc-outer-bound-zoc-n-links-w} and \ref{thm:mrc-inner-bound-zoc-n-links} to bound the maximum \gls{zoc}.}

\begin{lem}\label{lem:dist-cond-w}
	For $B(x_1,\dots{},x_n)=\sum_{i=1}^{n}F_{\X}(x_i)-n+1$, a continuous distributions $F_{\X}$ with a strictly quasi-concave density $f_{\X}$ is a $B$-\nice{} distribution, if the following sufficient condition holds
	\begin{equation}\label{eq:cond-nice-w}
		f_{\X}'\left(\inv{F_{\X}}\left(1-\frac{1}{n}\right)\right) < 0\,.
	\end{equation}
\end{lem}
\begin{proof}
	The proof can be found in Appendix~\ref{app:proof-lem-dist-cond-w}.
\end{proof}

\begin{lem}\label{lem:dist-cond-arch-copula}
	For $B(x_1,\dots{},x_n)=\sum_{i=1}^{n}\left(F_{\X}(x_i)\right)^{\frac{1}{n-1}}-n+1$, a continuous distribution $F_{\X}$ with a strictly quasi-concave density $f_{\X}$ is a $B$-\nice{} distribution, if the following two sufficient conditions hold.
	The function
	\begin{equation}
		g(x) = \frac{1}{n-1} F_{\X}(x)^{\frac{2-n}{n-1}}f_{\X}(x)
	\end{equation}
	is strictly quasi-concave, and
	\begin{equation}\label{eq:cond-nice-arch-copula}
		\frac{f_{\X}'(x^\star)}{\left(f_{\X}(x^\star)\right)^2} < \frac{n-2}{n-1} \left(1-\frac{1}{n}\right)^{1-n}%
	\end{equation}
	with
	\begin{equation}
		x^\star = \inv{F_{\X}}\left(\left(1-\frac{1}{n}\right)^{n-1}\right)\,.
	\end{equation}
\end{lem}
\begin{proof}
	The proof can be found in Appendix~\ref{app:proof-lem-dist-cond-arch-copula}.
\end{proof}

\begin{rem}%
	A stricter condition than \eqref{eq:cond-nice-arch-copula} is
	$f_{\X}'\left(\inv{F_{\X}}\left(\inv{\e}\right)\right) < 0$.
	Based on the unimodality of $f_{\X}$, this can also be written as $F_{\X}\left(\text{mode}\right) < \inv{\e}\approx 0.368$,
	where \enquote{mode} refers to the mode of distribution $F_{\X}$.
	Similarly, we can rewrite \eqref{eq:cond-nice-w} as
	$F_{\X}\left(\text{mode}\right) < 1-\frac{1}{n}$.
	Since these conditions are stricter than those from the above lemmas, they are also sufficient to characterize a $B$-\nice{} distribution.
\end{rem}

\begin{rem*}
Note that Lemmas~\ref{lem:dist-cond-w} and \ref{lem:dist-cond-arch-copula} only present sufficient conditions for $B$-\nice{} distributions.
It is, therefore, possible that other, possibly non-unimodal, distributions exist that are also $B$-\nice{}.
\end{rem*}

\begin{rem}[Common Fading Distributions]\label{rem:fading-dist-condition1}
	Most of the common fading distributions fulfill the sufficient conditions from Lemmas~\ref{lem:dist-cond-w} and \ref{lem:dist-cond-arch-copula}. In the case of Rayleigh fading, the channel gains $\X_i$ are exponentially distributed and therefore have a monotone density, i.e., $f'(x)<0$ for all $x$.
	Other common distributions like log-normal, Gamma, and $\chi^2$ also fulfill the conditions.
	An example of a distribution that is \emph{not} $B$-\nice{} for the considered $B$, is a Weibull distribution with scale parameter $\lambda=1$ and shape parameter $k=6$ for $n=2$. For this example, $f'(\inv{F}(0.5))\approx 1.98 > 0$.
	The mentioned examples are also illustrated in \cite{BesserGitlab}.
\end{rem}

\subsubsection{Outer Bounds on the Maximum ZOC}
First, we will give two outer bounds on the maximum \gls{zoc}. The first one is based on the Fr\'{e}chet-Hoeffding lower bound $W$. Even though it is not a valid copula for $n>2$, the bound from Theorem~\ref{thm:frechet-hoeffding-bounds} still holds and can therefore be used as a loose upper bound on the actual maximum \gls{zoc}.
\begin{thm}[Outer Bound on the Maximum \Gls{mrc}-\Gls{zoc} for $n$ Homogeneous Links based on $W$] \label{thm:mrc-outer-bound-zoc-n-links-w}
	Let $\X_1, \dots{}, \X_n$ be $n$ non-negative continuous random variables representing the channel gains of $n$ communication channels. All $\X_i$ follow the same distribution, i.e., $F_{\X_i}=F_{\X}$ for $i=1, \dots{}, n$, which is a $B$-\nice{} distribution for $B(x_1,\dots{},x_n)=\sum_{i=1}^{n}F_{\X}(x_i)-n+1$.
	The maximum zero-outage capacity $\zerorate$ for \gls{mrc} at the receiver is then upper bounded by
	\begin{equation}\label{eq:mrc-outer-bound-zoc-n-links-w}
		\overline{\overline{\zerorate_n}} = \log_2\left(1 + n\inv{F_{\X}}\left(1-\frac{1}{n}\right)\right)\,.
	\end{equation}
\end{thm}
\begin{proof}
	The proof can be found in Appendix~\ref{app:proof-thm-mrc-outer-w}.
\end{proof}

As already mentioned, $W$ is not a copula in the case of $n>2$. The outer bound from Theorem~\ref{thm:mrc-outer-bound-zoc-n-links-w} will therefore be loose.
In the following examples in Section~\ref{sub:rayleigh-mrc} and \ref{sub:nakagami-mrc}, we will see that the gap to the actual maximum \gls{zoc} can grow arbitrarily large for large $n$.
A better outer bound can be obtained by using \emph{joint mixability}~\cite{Wang2016}. Based on \cite[Thm.~2.6]{Wang2013}, a general bound on the best-case $\varepsilon$-outage capacity is derived in \cite[Thm.~7]{Besser2020twc}.
We will use this to derive an outer bound on the maximum \gls{zoc} in the following corollary.
\begin{cor}[{Outer Bound on the Maximum \Gls{mrc}-\Gls{zoc} for $n$ Homogeneous Links based on \cite[Thm.~7]{Besser2020twc}}] \label{cor:mrc-outer-bound-zoc-n-links-joint-mix}
	Let $\X_1, \dots{}, \X_n$ be $n$ non-negative continuous random variables representing the channel gains of $n$ communication channels. All $\X_i$ follow the same distribution, i.e., $F_{\X_i}=F_{\X}$ for $i=1, \dots{}, n$.
	The maximum zero-outage capacity $\zerorate$ for \gls{mrc} at the receiver is then upper bounded by
	\begin{equation}\label{eq:mrc-outer-bound-zoc-n-links-joint-mix}
		{\overline{\zerorate_n}} = \log_2\left(1 + n\expect{\X}\right)\,.
	\end{equation}
\end{cor}
\begin{proof}
	The proof directly follows from \cite[Thm.~7]{Besser2020twc}. We only need to set $\varepsilon=0$ and obtain \eqref{eq:mrc-outer-bound-zoc-n-links-joint-mix}.
\end{proof}

As stated in \cite{Wang2013}, the bound is tight, if the distribution $F_{\X}$ is $n$-completely mixable. Unfortunately, as shown in \cite[Rem.~2.2]{Wang2016}, distributions with a one-sided unbounded support can not be completely mixable. For details on this matter, we refer the readers to \cite{Wang2016,Puccetti2012}.
However, as shown in \cite{Besser2020twc}, the bound from Corollary~\ref{cor:mrc-outer-bound-zoc-n-links-joint-mix} might come arbitrarily close to the exact value in some special cases when $n\to\infty$.
We will also observe this behavior for the Rayleigh fading example in Section~\ref{sub:rayleigh-mrc}.

\subsubsection{Inner Bound on the Maximum ZOC}
Next, we derive an inner bound on the maximum \gls{zoc}.
It is based on the Archimedean copula stated in~\cite[Prop.~4.6]{McNeil2009}.
Archimedean copulas are a popular class of single-parameter copulas for an arbitrary number of dimensions $n$.
The simple construction for an arbitrary dimension $n>2$ is one of the reasons for their popularity~\cite[Chap.~4]{Nelsen2006}.
In the following, we need this extension to $n>2$ and since it is a valid copula, the derived value is achievable and we obtain an inner bound.
We use this particular copula since it is a lower bound on all Archimedean copulas~\cite[Prop.~4.6]{McNeil2009}.

Again, to obtain closed-form solutions, we assume that $F_{\X}$ is a $B$-\nice{} distribution.
However, with the mentioned copula, an inner bound can also be obtained for arbitrary fading distributions.

\begin{thm}[Inner Bound on the Maximum \Gls{mrc}-\Gls{zoc} for $n$ Homogeneous Links]\label{thm:mrc-inner-bound-zoc-n-links}
	Let $\X_1, \dots{}, \X_n$ be $n$ non-negative continuous random variables representing the channel gains of $n$ communication channels. All $\X_i$ follow the same distribution, i.e., $F_{\X_i}=F_{\X}$ for $i=1, \dots{}, n$, which is a $B$-\nice{} distribution for $B(x_1,\dots{},x_n)=\sum_{i=1}^{n}\left(F_{\X}(x_i)\right)^{\frac{1}{n-1}}-n+1$.
	The maximum zero-outage capacity $\zerorate$ for \gls{mrc} at the receiver is then lower bounded by
	\begin{equation}\label{eq:mrc-inner-bound-zoc-n-links}
		\underline{\zerorate_n} = \log_2\left(1 + n\inv{F_{\X}}\left(\left(1-\frac{1}{n}\right)^{n-1}\right)\right)\,.
	\end{equation}
\end{thm}
\begin{proof}
	The proof can be found in Appendix~\ref{app:proof-thm-mrc-inner}.
\end{proof}

\subsubsection{Gap Between Inner and Outer Bound}
For distributions that are $B$-\nice{} distributions for the boundary function in Theorem~\ref{thm:mrc-inner-bound-zoc-n-links}, we know that the exact value of the maximum \gls{zoc} is between the outer bound given in Corollary~\ref{cor:mrc-outer-bound-zoc-n-links-joint-mix} and the inner bound from Theorem~\ref{thm:mrc-inner-bound-zoc-n-links}.
It is therefore of interest to analyze the gap between the bounds.
This will be summarized in the following corollary.
\begin{cor}[Maximum Gap between Inner and Outer Bound on the Maximum \Gls{mrc}-\Gls{zoc}]\label{cor:mrc-max-gap-inner-outer-bounds-n-links}
	Let $\X_1, \dots{}, \X_n$ be $n$ non-negative continuous random variables representing the channel gains of $n$ communication channels. All $\X_i$ follow the same distribution, i.e., $F_{\X_i}=F_{\X}$ for $i=1, \dots{}, n$, which is $B$-\nice{} for the function $B$ from Theorem~\ref{thm:mrc-inner-bound-zoc-n-links}.
	The gap between the inner bound on the maximum \gls{zoc} from Theorem~\ref{thm:mrc-inner-bound-zoc-n-links} and the outer bound from Corollary~\ref{cor:mrc-outer-bound-zoc-n-links-joint-mix} is at most
	\begin{equation}\label{eq:mrc-max-gap-inner-outer-bounds-n-links}
		\overline{\zerorate_n} - \underline{\zerorate_n} \leq \log_2\left(\frac{\expect{\X}}{\inv{F_{\X}}(\inv{\e})}\right)\,.
	\end{equation}
\end{cor}
\begin{proof}
	The proof can be found in Appendix~\ref{app:proof-thm-mrc-max-gap}.
\end{proof}
Corollary~\ref{cor:mrc-max-gap-inner-outer-bounds-n-links}, therefore, allows us to calculate the maximum \gls{zoc} for $n$ homogeneous fading links and \gls{mrc} at the receiver within a finite amount of bits equal to $\log_2\left({\expect{\X}}\right) - \log_2\left({\inv{F_{\X}}(1/\e)}\right)$.

\subsection{Example: Rayleigh Fading}\label{sub:rayleigh-mrc}
In the following, we will illustrate the general results with the example of Rayleigh fading. In this case, all channel gains $\abs{{H}_i}^2$ are exponentially distributed with mean $1$. This gives $\snr_i\abs{{H}_i}^2=\X_i\sim\exp(1/\snr_i)$.
The \gls{cdf} and inverse \gls{cdf} of $\X_i\sim\exp(\lambda_i)$ are given by
\begin{equation*}
	F_{\X_i}(x) = \begin{cases}
		0 & \text{if } x<0\\
		1-\exp\left(-\lambda_i x\right) & \text{if } x\geq 0
	\end{cases}%
\end{equation*}
and
\begin{equation*}
	\inv{F_{\X_i}}(u) = \begin{cases}
		\frac{-\log(1-u)}{\lambda_i} & \text{if } 0\leq u < 1\\
		+\infty & \text{if } u = 1\,,
	\end{cases}%
\end{equation*}
respectively. Note that the expected value of $\X_i$ is $1/\lambda_i=\snr_i$.

\subsubsection{Two Links}
First, we will take a look at the two-dimensional case.
In order to do this, we will evaluate \eqref{eq:mrc-zero-outage-capacity-two-links} from Theorem~\ref{thm:mrc-achievable-zero-out-two-links} in the following.

We start with determining $x^\star$ according to \eqref{eq:condition-x-star}. This gives the following expression
\begin{equation*}
	x^\star(t) = \frac{-1}{\lx}\log\left(\frac{\ly}{\lx+\ly}(2-t)\right)\,.
\end{equation*}
Note that the range of $x^\star$ is bounded by \num{0} and $\inv{F_{\X_1}}(t)$.
The boundary of $\mathcal{B}$ is computed according to \eqref{eq:def-b} as
\begin{equation*}
	B_t(x) = -\frac{\log\left(2-t-\exp(-\lx x)\right)}{\ly}\,,
\end{equation*}
and therefore, we get
\begin{equation}
	B_t(x^\star) = -\frac{\log\left((2-t)\frac{\lx}{\lx+\ly}\right)}{\ly}\,,
\end{equation}
when $0<x^\star<\inv{F_{\X_1}}(t)$. For the extreme cases, we find that
\begin{equation*}
	x^\star+B(x^\star) = \begin{cases}
		\inv{F_{\X_2}}(t) & \text{if } x^\star = 0\\
		\inv{F_{\X_1}}(t) & \text{if } x^\star = \inv{F_{\X_1}}(t)\\
	\end{cases}\,.
\end{equation*}
Thus, we can combine the above results according to \eqref{eq:mrc-zero-outage-capacity-two-links} to get the expression of the \gls{zoc} for two Rayleigh fading links as
\begin{equation}\label{eq:mrc-rayleigh-zero-cap}
	\zerorate(t) = \log_2\left(1 + x^\star(t) -\frac{\log\left(2-t-\exp(-\lx x^\star(t))\right)}{\ly}\right)
\end{equation}
with
\begin{equation}\label{eq:rayleigh-x-star}
	x^\star(t) = \positive{\min\left\{\frac{-1}{\lx}\log\left(\frac{\ly}{\lx+\ly}(2-t)\right), \inv{F_{\X_1}}(t)\right\}}.
\end{equation}

Figure~\ref{fig:rayleigh-boundary-t} shows examples for the function $B_t$ and the possible candidates for $s=2^{\zerorate}-1$. Recall that the idea is to find the line $x_1+x_2=s$ with the maximum $s$ such that the line is still below $B_t$.
The different curves are shown for two values of $t$, namely $t=0.5$ and $t=0.9$.
It can be seen that the values $\inv{F_{\X_1}}(t)$ and $\inv{F_{\X_2}}(t)$ increase with increasing $t$.
These are the first candidates for the optimal $\zerorate$ from \eqref{eq:mrc-zero-outage-capacity-two-links}.
They are represented by the lines $x_1+x_2=\inv{F_{\X_1}}(t)$ and $x_1+x_2=\inv{F_{\X_2}}(t)$, respectively.
In the case of $t=0.5$, the optimal $s$ is given by $\inv{F_{\X_2}}(0.5)$ since there is no other line of the form $x_1+x_2=s$ with a larger $s$ which is still below $B_t(x)$.
In contrast, the line $x_1+x_2=\inv{F_{\X_1}}(t)$ is not below $B_t$.
For the larger value of $t=0.9$, there exists a tangent point at around $x^\star=0.18$ which gives the maximum $s=2^{\zerorate}-1$ of around \num{0.6}.
This can then be used to determine the \gls{zoc} $\zerorate$.
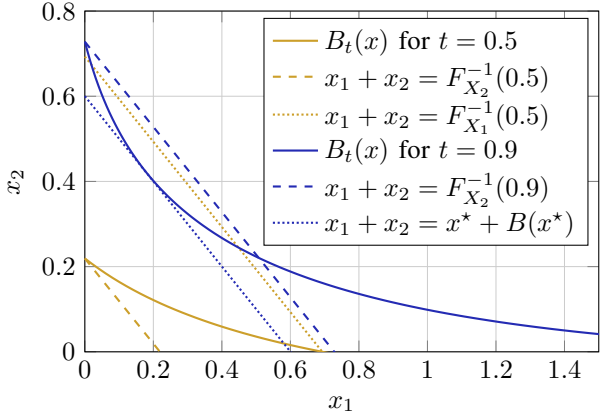
\begin{figure}
	\centering
	\begin{tikzpicture}
\begin{axis}[
	width=.95\linewidth,
	height=.25\textheight,
	xlabel={$x_1$},
	ylabel={$x_2$},
	domain=0:2,
	xmin=0,
	xmax=1.5,
	ymin=0,
	ymax=.8,
	legend cell align=left,
	ymajorgrids,
	xmajorgrids,
	xminorgrids,
	grid style={line width=.1pt, draw=gray!20},
	major grid style={line width=.25pt,draw=gray!40},
]

\addplot[plot3, thick, samples=100] {-ln(2-0.5-exp(-1*x))/3.162};
\addlegendentry{$B_t(x)$ for $t=0.5$};

\addplot[plot3, dashed, thick] {0.21919-x};
\addlegendentry{$x_1+x_2 = \inv{F_{\X_2}}(0.5)$};

\addplot[plot3, densely dotted, thick] {0.693147-x};
\addlegendentry{$x_1+x_2 = \inv{F_{\X_1}}(0.5)$};

\addplot[plot1, thick, samples=250] {-ln(2-0.9-exp(-1*x))/3.162};
\addlegendentry{$B_t(x)$ for $t=0.9$};
\addplot[plot1, dashed, thick] {0.7281-x};
\addlegendentry{$x_1+x_2 = \inv{F_{\X_2}}(0.9)$};
\addplot[plot1, densely dotted, thick] {0.60028-x};
\addlegendentry{$x_1+x_2 = x^\star+B(x^\star)$};

\end{axis}
\end{tikzpicture}
	\caption{Boundary $B_t(x)$ for Rayleigh fading channels with \glspl{snr} $\snrx=\SI{0}{\decibel}$ and $\snry=\SI{-5}{\decibel}$ for the values $t=0.5$ and $t=0.9$.}
	\label{fig:rayleigh-boundary-t}
\end{figure}

The \gls{zoc} $\zerorate$ from \eqref{eq:mrc-rayleigh-zero-cap} is shown for different values of $\snrx$ and $\snry$ in Fig.~\ref{fig:rayleigh-zero-cap}.
As expected, the \gls{zoc} increases for increasing \gls{snr} values $\rho_i$, since the channel quality increases.
An interesting phenomenon can be seen from the asymmetric constellations of $\snrx$ and $\snry$.
Especially, if there is a big difference between them, e.g., $\snrx=\SI{-5}{\decibel}$ and $\snry=\SI{10}{\decibel}$, the \gls{zoc} is low and grows only slowly for small $t$.
However, for larger $t$, the rate of growth increases.
The reason for this is that $\zerorate$ is only determined by the weaker channel for small $t$.
This can easily be seen when comparing the case of $\snrx=\SI{-5}{\decibel}$ and $\snry=\SI{10}{\decibel}$ with the case of $\snrx=\SI{-5}{\decibel}$ and $\snry=\SI{5}{\decibel}$.
For $t$ up to around \num{0.9}, both cases achieve the same \gls{zoc} $\zerorate$, since their weaker channel has the same \gls{snr} of \SI{-5}{\decibel}.
However, for $t>0.9$, the constellation with the better second channel, i.e., higher $\snry$, is able to achieve a larger $\zerorate$.
The same behavior can be seen for the cases of $\snrx=\snry=\SI{0}{\decibel}$ and $\snrx=\SI{0}{\decibel}$, $\snry=\SI{5}{\decibel}$.
The only difference is that the value of $t$ at which the curves start to differ is lower, at around \num{0.1}.
This can also be seen from Fig.~\ref{fig:rayleigh-grid-snr}, where the \glspl{zoc} of combinations of $\snrx$ and $\snry$ are shown for $t=0.5$.
All of the presented figures are also available as interactive versions at \cite{BesserGitlab}.
We encourage the interested readers to change the parameters on their own and explore the behavior of the presented results.
\begin{figure}[!t]
	\centering
	\subfloat[{\Gls{zoc} for different copula parameters $t$.}]{\begin{tikzpicture}%
\begin{axis}[
	width=.95\linewidth,
	height=.25\textheight,
	xlabel={Copula Parameter $t$},
	ylabel={$\zerorate(t)$},
	xmin=0,
	xmax=1,
	ymin=0,
	legend pos=north west,
	legend style={font=\small}, %
	ymajorgrids,
	xmajorgrids,
	xminorgrids,
	grid style={line width=.1pt, draw=gray!20},
	major grid style={line width=.25pt,draw=gray!40},
]
	\addplot[plot0,thick,mark=*, mark repeat=3] table[x=t,y=capac] {data/zero-out-capac-rayleigh-snrx0.0-snry0.dat};
	\addlegendentry{$\snrx=\snry=\SI{0}{\decibel}$};
	
	\addplot[plot1,thick,mark=pentagon*, mark repeat=3, mark phase=0] table[x=t,y=capac] {data/zero-out-capac-rayleigh-snrx5.0-snry5.0.dat};
	\addlegendentry{$\snrx=\snry=\SI{5}{\decibel}$};
	
	\addplot[plot2,thick,mark=square*, mark repeat=3, mark phase=0] table[x=t,y=capac] {data/zero-out-capac-rayleigh-snrx0.0-snry5.0.dat};
	\addlegendentry{$\snrx=\SI{0}{\decibel}, \snry=\SI{5}{\decibel}$};
	
	\addplot[plot3,thick,mark=triangle*, mark repeat=3, mark phase=0] table[x=t,y=capac] {data/zero-out-capac-rayleigh-snrx-5.0-snry5.0.dat};
	\addlegendentry{$\snrx=\SI{-5}{\decibel}, \snry=\SI{5}{\decibel}$};
	
	\addplot[plot4,thick,mark=diamond*, mark repeat=3, mark phase=0] table[x=t,y=capac] {data/zero-out-capac-rayleigh-snrx-5.0-snry10.0.dat};
	\addlegendentry{$\snrx=\SI{-5}{\decibel}, \snry=\SI{10}{\decibel}$};
\end{axis}
\end{tikzpicture}
		\label{fig:rayleigh-zero-cap}
	}

	\subfloat[{The copula parameter is set to $t=0.5$. The marked points correspond to the \gls{snr} combinations shown in Fig.~\ref{fig:rayleigh-zero-cap}.}]{\begin{tikzpicture}
\begin{axis}[
	width=.95\linewidth,
	height=.25\textheight,
	xlabel={\Gls{snr} $\snrx$ [dB]},
	ylabel={\Gls{snr} $\snry$ [dB]},
	view={0}{90},
	colormap name=viridis,
	point meta min=0,
	xmin=-10,
	xmax=10,
	ymin=-10,
	ymax=10,
	mark size=3,
	ymajorgrids,
	xmajorgrids,
	xminorgrids,
	grid style={line width=.1pt, draw=gray!20},
	major grid style={line width=.25pt,draw=gray!40},
]
    \addplot3[contour gnuplot={number=20, labels=true},
    		  mesh/rows=50,
    		  mesh/cols=50,
    		  patch type=bilinear,
    		 ] table {data/grid-zero-out-snr-t0.5.dat};
    \addplot[plot0, mark=*] coordinates {(0, 0)};
    \addplot[plot1, mark=pentagon*] coordinates {(5,5)};
    \addplot[plot2, mark=square*] coordinates {(0, 5)};
    \addplot[plot3, mark=triangle*] coordinates {(-5, 5)};
    \addplot[plot4, mark=diamond*] coordinates {(-5, 10)};
\end{axis}
\end{tikzpicture}
		\label{fig:rayleigh-grid-snr}
	}
	\caption{Achievable zero-outage capacities $\zerorate$ for two dependent Rayleigh fading channels with different \gls{snr} values $\snrx$ and $\snry$. The two channels follow the copula $C_t$.}\label{fig:rayleigh-mrc-copula-t}
\end{figure}
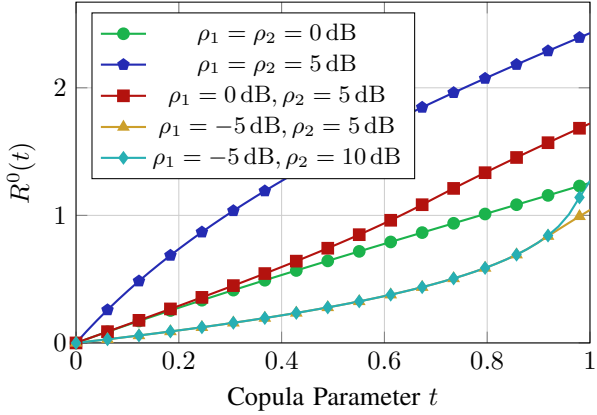
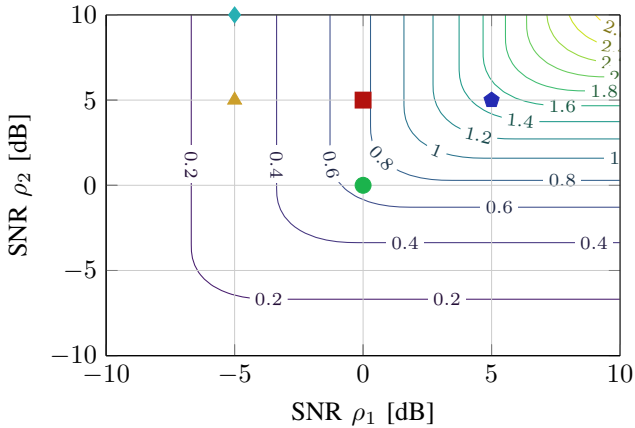

The maximum \gls{zoc} is achieved for countermonotonic $\X_1$ and $\X_2$. This corresponds to a value of $t=1$.
The joint distribution $F_{\X, \Y}$ in this case is supported on the line
\begin{equation*}
	x_2 = \inv{F_{\X_2}}(1-F_{\X_1}(x_1)) = \frac{-\log\left(1-\exp(-\lx x_1)\right)}{\ly}\,.
\end{equation*}
From \eqref{eq:rayleigh-x-star}, we get
$x^\star(1) = -\left(\log\frac{\ly}{\lx+\ly}\right)/{\lx}$
and together with \eqref{eq:mrc-rayleigh-zero-cap} this yields the maximum \gls{zoc}
\begin{equation}
	\zerorate(1) = \log_2\left(1 + \frac{\log\frac{\lx+\ly}{\ly}}{\lx} + \frac{\log\frac{\lx+\ly}{\lx}}{\ly}\right)\,.
\end{equation}
In the case that $\lx=\ly=1$, this evaluates to $\zerorate(1) = \log_2(1+2\log 2)$. The result for this special case was also derived in \cite[Thm.~1]{Jorswieck2019} and \cite[Ex.~3]{Besser2020twc} by different approaches.

\subsubsection{General Case}
For the extension to the $n$-dimensional case, we now assume homogeneous links, i.e., $F_{\X_1}=\cdots{}=F_{\X_n}=F_{\X}$. We can then apply the results from Section~\ref{sub:mrc-general-n-links} to find bounds on the maximum \gls{zoc}.
In addition, if $F_{\X}$ has a strictly monotone density, the exact values for the maximum \gls{zoc} are derived in \cite{Besser2020twc}. An example with this property is Rayleigh fading, since the channel gains $\X_i$ are exponentially distributed.
In the following example, we can therefore compare the exact values from \cite{Besser2020twc} to the bounds from Theorems~\ref{thm:mrc-outer-bound-zoc-n-links-w}, \ref{thm:mrc-inner-bound-zoc-n-links}, and Corollary~\ref{cor:mrc-outer-bound-zoc-n-links-joint-mix} for Rayleigh fading.
The results are shown in Fig.~\ref{fig:mrc-rayleigh-max-zoc-n-links}.

The first upper bound $\overline{\overline{\zerorate_n}}$, based on the Fr\'{e}chet-Hoeffding lower bound $W$, is evaluated according to \eqref{eq:mrc-outer-bound-zoc-n-links-w}.
The second upper bound from Corollary~\ref{cor:mrc-outer-bound-zoc-n-links-joint-mix} is calculated to
\begin{equation}
	\overline{\zerorate_n} = \log_2\left(1 + n\snr\right)\,.
\end{equation}
As shown in \cite{Besser2020twc}, the exact value approaches $\overline{\zerorate_n}$ for $n\to\infty$ in the case of Rayleigh fading.
The inner bound from Theorem~\ref{thm:mrc-inner-bound-zoc-n-links}, is evaluated to
\begin{equation}
	\underline{\zerorate_n} = \log_2\left(1 - \snr n \log\left(1 - \left(1 - \frac{1}{n}\right)^{n-1}\right)\right)\,.
\end{equation}
The exact value of the \gls{zoc} $\zerorate_n$ is between the inner and outer bound, i.e., $\underline{\zerorate_n}\leq \zerorate_n \leq \overline{\zerorate_n}$.
As described in Corollary~\ref{cor:mrc-max-gap-inner-outer-bounds-n-links}, the gap between the bounds is upper bounded by
\begin{equation}
	\lim\limits_{n\to\infty} \overline{\zerorate_n}-\underline{\zerorate_n} = -\log_2\left(1 - \log(\e-1)\right) \approx \SI{1.12}{bit}.
\end{equation}

\begin{figure}
	\centering
	\begin{tikzpicture}%
\begin{axis}[
	width=.98\linewidth,
	height=.25\textheight,
	xlabel={Number of Links $n$},
	ylabel={Maximum ZOC ${\zerorate_n}$},
	xmin=2,
	xmax=10,
	legend pos=north west,
	ymajorgrids,
	xmajorgrids,
	xminorgrids,
	grid style={line width=.1pt, draw=gray!20},
	major grid style={line width=.25pt,draw=gray!40},
]
	\addplot[plot0,thick,mark=*] table[x=n,y=exact] {data/bounds_zoc_rayleigh_symmetric.dat};
	\addlegendentry{Exact Values from \cite{Besser2020twc}};
	
	\addplot[plot1,thick,mark=pentagon*] table[x=n,y=fhW] {data/bounds_zoc_rayleigh_symmetric.dat};
	\addlegendentry{Outer Bound from Theorem~\ref{thm:mrc-outer-bound-zoc-n-links-w}};
	
	\addplot[plot3,thick,mark=square*] table[x=n,y=outerPhi] {data/bounds_zoc_rayleigh_symmetric.dat};
	\addlegendentry{Outer Bound from Corollary~\ref{cor:mrc-outer-bound-zoc-n-links-joint-mix}};
	
	\addplot[plot2,thick,mark=triangle*] table[x=n,y=inner] {data/bounds_zoc_rayleigh_symmetric.dat};
	\addlegendentry{Inner Bound from Theorem~\ref{thm:mrc-inner-bound-zoc-n-links}};
\end{axis}
\end{tikzpicture}
	\caption{Exact values and bounds on the maximum \gls{zoc} for $n$ Rayleigh fading links with $\snr=\SI{0}{\decibel}$.}
	\label{fig:mrc-rayleigh-max-zoc-n-links}
\end{figure}
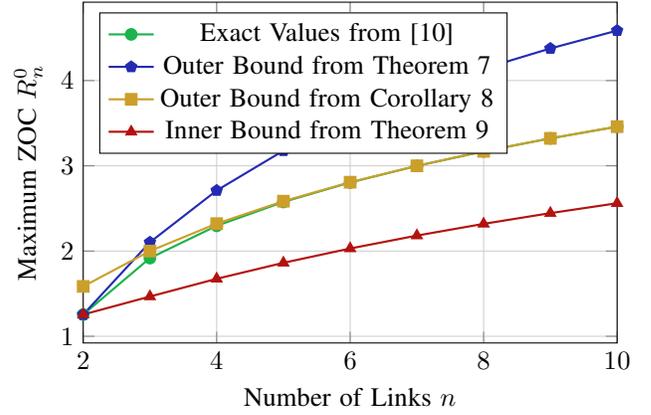

We want to emphasize that the exact values derived in \cite{Besser2020twc} only hold for monotonic densities of the channel gains $\X_i$. It can therefore not be used for all types of fading, e.g., it can not be used for log-normal fading.
Therefore the upper and lower bounds derived in the present work are useful and more general.

\subsection{\texorpdfstring{Example: Nakagami-$m$ Fading}{Example: Nakagami-m Fading}}\label{sub:nakagami-mrc}
We will now give an example of Nakagami-$m$ fading.
In this case, $\abs{{H}_i}$ is distributed according to a Nakagami-$m$ distribution, i.e., $\abs{{H}_i}\sim\text{Nakagami}(m, 1)$~\cite{Yacoub2002}.
Therefore, the channel gains $\X_i=\snr_i\abs{{H}_i}^2$ are distributed according to a Gamma-distribution, i.e., $\X_i\sim\Gamma(m, \snr_i/m)$.

The \glspl{zoc} are calculated numerically according to Theorem~\ref{thm:mrc-achievable-zero-out-two-links}. The source code of the calculations together with interactive versions of the presented figures can be found in \cite{BesserGitlab}.
Figure~\ref{fig:mrc-naka-zoc-copula-t} shows $\zerorate(t)$ for two Nakagami-$m$ fading channels with $m=5$ for different \gls{snr} combinations. As expected, the achievable \gls{zoc} increases with increasing \glspl{snr}.
For asymmetric \gls{snr} values $\snrx$ and $\snry$, we can observe a similar behavior compared to the case of Rayleigh fading discussed in the previous section. Up to a certain value of $t$, the achievable \gls{zoc} is only determined by the worse channel. In Fig.~\ref{fig:mrc-naka-zoc-copula-t} this can be seen, e.g., for $\snrx=\SI{0}{\decibel}$ and $\snry=\SI{0}{\decibel}$ or $\SI{5}{\decibel}$. Up to $t$ around \num{0.9}, both constellations achieve the same \gls{zoc}. Above this value, the scenario with the higher $\snry$ achieves higher $\zerorate$.
We encourage the interested readers to explore this behavior in the interactive version in \cite{BesserGitlab} with further parameter constellations.

For the homogeneous case with $n>2$, we select the parameters $m=5$ and $\snr=\SI{0}{\decibel}$.
The Gamma distribution with these parameters fulfills the inequality that the mode is less than the median,
\begin{equation*}
	\mode = \frac{m-1}{m}\snr = 0.8 < \median = 0.934\,,
\end{equation*}
and we can therefore apply Theorem~\ref{thm:mrc-outer-bound-zoc-n-links-w}.
It is also straightforward to confirm that the distribution fulfills the condition from Lemma~\ref{lem:dist-cond-arch-copula}, which allows us to use Theorem~\ref{thm:mrc-inner-bound-zoc-n-links}.
The results for this example are shown in Fig.~\ref{fig:mrc-nakagami-max-zoc-n-links}.
Similar to the case of Rayleigh fading, the gap between the inner bound $\underline{\zerorate_n}$ and the outer bound $\overline{\overline{\zerorate_n}}$ from Theorem~\ref{thm:mrc-outer-bound-zoc-n-links-w} grows indefinitely with increasing $n$.
On the other hand, the gap between the inner bound $\underline{\zerorate_n}$ and the outer bound $\overline{\zerorate_n}$ from Corollary~\ref{cor:mrc-outer-bound-zoc-n-links-joint-mix} is, according to Corollary~\ref{cor:mrc-max-gap-inner-outer-bounds-n-links}, always less than around \SI{0.328}{bit}.
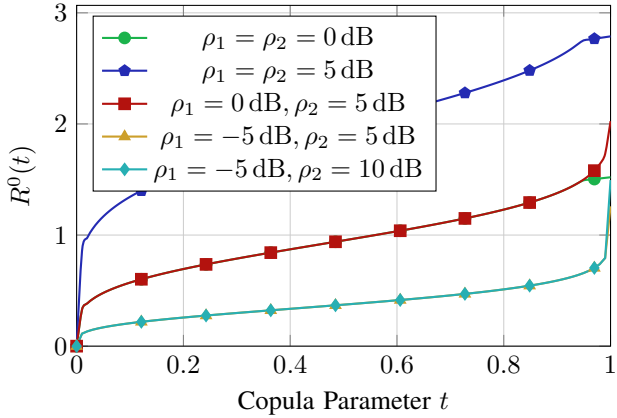
\begin{figure}
	\centering
	\begin{tikzpicture}%
\begin{axis}[
	width=.98\linewidth,
	height=.25\textheight,
	xlabel={Copula Parameter $t$},
	ylabel={$\zerorate(t)$},
	xmin=0,
	xmax=1,
	ymin=0,
	legend pos=north west,
	ymajorgrids,
	xmajorgrids,
	xminorgrids,
	grid style={line width=.1pt, draw=gray!20},
	major grid style={line width=.25pt,draw=gray!40},
]
	\addplot[plot0,thick,mark=*, mark repeat=12, smooth] table[x=t,y=capac] {data/zoc-x_naka5-y_naka5-snrx0-snry0.dat};
	\addlegendentry{$\snrx=\snry=\SI{0}{\decibel}$};
	
	\addplot[plot1,thick,mark=pentagon*, mark repeat=12, mark phase=0, smooth] table[x=t,y=capac] {data/zoc-x_naka5-y_naka5-snrx5.0-snry5.0.dat};
	\addlegendentry{$\snrx=\snry=\SI{5}{\decibel}$};
	
	\addplot[plot2,thick,mark=square*, mark repeat=12, mark phase=0, smooth] table[x=t,y=capac] {data/zoc-x_naka5-y_naka5-snrx0.0-snry5.0.dat};
	\addlegendentry{$\snrx=\SI{0}{\decibel}, \snry=\SI{5}{\decibel}$};
	\addplot[plot3,thick,mark=triangle*, mark repeat=12, mark phase=0] table[x=t,y=capac] {data/zoc-x_naka5-y_naka5-snrx-5.0-snry5.0.dat};
	\addlegendentry{$\snrx=\SI{-5}{\decibel}, \snry=\SI{5}{\decibel}$};
	\addplot[plot4,thick,mark=diamond*, mark repeat=12, mark phase=0] table[x=t,y=capac] {data/zoc-x_naka5-y_naka5-snrx-5.0-snry10.0.dat};
	\addlegendentry{$\snrx=\SI{-5}{\decibel}, \snry=\SI{10}{\decibel}$};
\end{axis}
\end{tikzpicture}
	\caption{Achievable zero-outage capacities $\zerorate$ for two dependent Nakagami-$m$ fading channels with different \gls{snr} values $\snrx$ and $\snry$ and $m=5$. The two channels follow the copula $C_t$ defined in Remark~\ref{rem:ambiguity-copula-t}.
	}
	\label{fig:mrc-naka-zoc-copula-t}
\end{figure}
\begin{figure}
	\centering
	\begin{tikzpicture}%
\begin{axis}[
	width=.98\linewidth,
	height=.25\textheight,
	xlabel={$n$},
	ylabel={Maximum ${\zerorate_n}$},
	xmin=2,
	xmax=10,
	legend pos=north west,
	ymajorgrids,
	xmajorgrids,
	xminorgrids,
	grid style={line width=.1pt, draw=gray!20},
	major grid style={line width=.25pt,draw=gray!40},
]	
	\addplot[plot1,thick,mark=pentagon*] table[x=n,y=zocOuterW] {data/zoc-mrc-naka5-snr0.dat};
	\addlegendentry{Outer Bound from Theorem~\ref{thm:mrc-outer-bound-zoc-n-links-w}};
	
	\addplot[plot3,thick,mark=square*] table[x=n,y=zocOuterJM] {data/zoc-mrc-naka5-snr0.dat};
	\addlegendentry{Outer Bound from Corollary~\ref{cor:mrc-outer-bound-zoc-n-links-joint-mix}};
	
	\addplot[plot2,thick,mark=triangle*] table[x=n,y=zocInner] {data/zoc-mrc-naka5-snr0.dat};
	\addlegendentry{Inner Bound from Theorem~\ref{thm:mrc-inner-bound-zoc-n-links}};
\end{axis}
\end{tikzpicture}
	\caption{Inner and outer bounds on the maximum \gls{zoc} for $n$ Nakagami-$m$ fading links with $m=5$ and $\snr=\SI{0}{\decibel}$.}
	\label{fig:mrc-nakagami-max-zoc-n-links}
\end{figure}
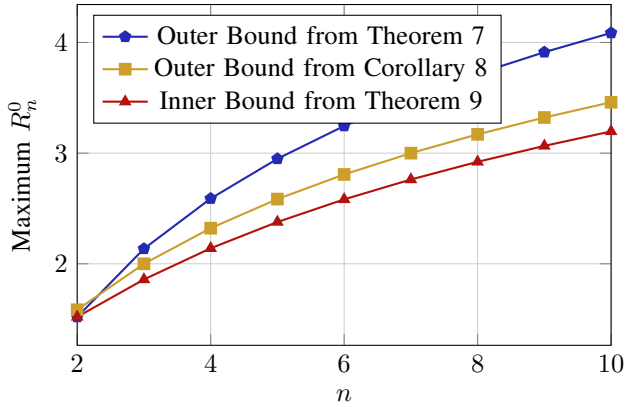

\subsection{Summary}
In this section, we showed that there exists an infinite number of joint fading distributions for which the \gls{zoc} is positive, if the receiver employs \gls{mrc}.
For the two-dimensional scenario, we derived a general expression for the maximum \gls{zoc}, which can also be applied to heterogeneous marginals.
For homogeneous marginals, we extended the results to the case of $n>2$ under some constraints and provide upper and lower bounds on the maximum \gls{zoc} which admit a finite gap.
An overview of the results can be found in Table~\ref{tab:mrc-max-zoc-overview}.
The most general scenario of $n$ heterogeneous marginals remains an open problem at this point.
\begin{table*}[!htb]
	\renewcommand{\arraystretch}{1.25}
	\centering
	\caption{Maximum \Gls{zoc} for Maximum Ratio Combining $L_{\text{MRC}}(\X_1, \dots{}, \X_n) = \sum_{i=1}^{n} \X_i$}
	\label{tab:mrc-max-zoc-overview}
	\begin{tabular}{c|cc|c}
		\toprule
		\multirow{2}{*}{\textbf{}} & \multicolumn{2}{c|}{\textbf{Homogeneous Links}} & \multirow{2}{*}{\textbf{Heterogeneous Links}}\\
		& {\textbf{\boldmath $B$-\nice{} Distribution}} & \textbf{General} & \\
		\midrule
		$\mathbf{n=2}$ & $\displaystyle \log_2\left(1 + 2\median(\X)\right)$ & See Heterogeneous Links & $\displaystyle \log_2\left(1 + x^\star + \inv{F_{\X_2}}(1-F_{\X_1}(x^\star))\right)$\\
		$\mathbf{n>2}$ & $\displaystyle \overline{\zerorate_n} - \underline{\zerorate_n} \leq \log_2\left(\frac{\expect{\X}}{\inv{F_{\X}}\left(\frac{1}{\e}\right)}\right)$ & Open Problem & Open Problem\\
		\bottomrule
	\end{tabular}
\end{table*}
\section{Selection Combining}\label{sec:sc}

In this section, we will consider \gls{sc} as another diversity combining scheme. In this case, only the strongest link is selected at the receiver~\cite{Brennan1959}. The combining function $L$ is therefore given as
\begin{equation*}
	L_{\text{SC}}(\X_1, \dots{}, \X_n) = \max \left\{\X_1, \dots{}, \X_n\right\}\,.
\end{equation*}

Analogue to Section~\ref{sec:mrc}, we start with the two-dimensional case and extend it to the general $n$-dimensional case. We also evaluate the results for the examples of Rayleigh fading and Nakagami-$m$ fading.

\subsection{Two-Dimensional Case}
Similar to \gls{mrc}, the maximum \gls{zoc} for \gls{sc} with two dimensions is achieved for countermonotonic $\X_1$ and $\X_2$~\cite{Frank1987}.
The optimization problem \eqref{eq:opt-problem-boundary} can therefore be written as
\begin{equation}
	\min_{F_{\X_1}(x_1)+F_{\X_2}(x_2)=1} \max \left\{x_1, x_2\right\}\,.
\end{equation}
With the substitutions $x_1 = \inv{F_{\X_1}}(p)$ and $x_2 = \inv{F_{\X_2}}(1-p)$, we rewrite the problem as
\begin{equation}
	\min_{p\in[0, 1]} \max \left\{\inv{F_{\X_1}}(p), \inv{F_{\X_2}}(1-p)\right\}\,.
\end{equation}
Recall that the quantile function $\inv{F_{\X}}$ is an increasing function. Combined with the assumption that all $\X_i$ are supported on the non-negative real numbers, we can derive that the minimum is attained at $p^\star$ for which
\begin{equation}
	\inv{F_{\X_1}}(p^\star) = \inv{F_{\X_2}}(1-p^\star)
\end{equation}
holds.
The maximum \gls{zoc} is then given as
\begin{equation}\label{eq:sc-max-zoc-2links-hetero}
	\zerorate = \log_2\left(1 + \inv{F_{\X_1}}(p^\star)\right)\,.
\end{equation}

In the case of homogeneous marginals, i.e., $F_{\X_1}=F_{\X_2}=F_{\X}$, we get $p^\star=0.5$. The \gls{zoc} from \eqref{eq:sc-max-zoc-2links-hetero} can then be simplified to $\zerorate = \log_2\left(1 + \median(\X)\right)$.

\subsection{General Case}
For the extension to the $n$-dimensional case, first, recall the general problem formulation from~\eqref{eq:def-prob-eps-outage-capac}.
We are interested in the case that the outage probability is zero, i.e., $\Pr(L(\X_1, \dots{}, \X_n)<s)=0$.

For homogeneous marginals $F_{\X_1}=\cdots{}=F_{\X_n}=F_{\X}$, we determine the exact solution to \eqref{eq:def-prob-eps-outage-capac} as follows.
From \cite{Lai1978}, we know that there always exists a \emph{maximally dependent} coupling of $(\X_1, \dots{}, \X_n)$ such that
\begin{equation}
	\Pr\left(\max (\X_1, \dots{}, \X_n) < \inv{F_{\X}}\left(1-\frac{1}{n}\right)\right) = 0\,,
\end{equation}
if all $\X_i$ are identically distributed.
This can now be used to solve the optimization problem~\eqref{eq:def-prob-eps-outage-capac} and find the maximum \gls{zoc} for \gls{sc} at the receiver as
\begin{equation}\label{eq:sc-max-zoc-n-links-homo}
	\zerorate_n = \log_2\left(1 +\inv{F_{\X}}\left(1-\frac{1}{n}\right)\right)\,.
\end{equation}

In order to extend the above to the general $n$-dimensional case with heterogeneous marginal distributions, we leverage the results derived in \cite{Ruschendorf1980}.
With \cite[Thm.~7]{Ruschendorf1980}, we derive the implicit characterization of the maximum \gls{zoc} $\zerorate_n$ for \gls{sc} as
\begin{equation}\label{eq:sc-max-zoc-n-links-hetero}
	\sum_{i=1}^{n} F_{\X_i}\left(2^{\zerorate_n} - 1\right) = n-1\,,
\end{equation}
where $F_{\X_i}$ corresponds to marginal distributions of $\X_i$, i.e., $\X_i\sim F_{\X_i}$.
Note that this is consistent with the result from \eqref{eq:sc-max-zoc-n-links-homo} for homogeneous marginals, i.e., when $F_{\X_1} = \cdots{} = F_{\X_n} = F_{\X}$.

\subsection{Examples}\label{sub:sc-examples}
We will now evaluate the results for \gls{sc} at the receiver for Rayleigh fading and Nakagami-$m$ fading.
The first example is for two heterogeneous links. The first link is Rayleigh fading, i.e., $\X_1\sim\exp(1/\snrx)$, while the second is Nakagami-$m$ fading, i.e., $\X_2\sim\Gamma(m, \snry/m)$.
The parameters are set to $m=5$ and $\snrx=\snry=\SI{10}{\decibel}$.
The maximum \gls{zoc} is evaluated according to \eqref{eq:sc-max-zoc-2links-hetero}, which yields $p^\star=0.575$, $s^\star_{\text{SC}}=8.554$, and $\zerorate=3.256$.

Next, we consider the homogeneous case with $n\geq 2$. In this case, the maximum \gls{zoc} is given by \eqref{eq:sc-max-zoc-n-links-homo}. For Rayleigh fading, we get
\begin{equation}\label{eq:sc-max-zoc-n-links-rayleigh}
	\zerorate_n = \log_2\left(1 + \snr\log n\right)\,,
\end{equation}
and for Nakagami-$m$ fading, we get
\begin{equation}\label{eq:sc-max-zoc-n-links-nakagami}
	\zerorate_n = \log_2\left(1 + \frac{\snr}{m} \inv{P}\left(m, 1-\frac{1}{n}\right)\right)\,,
\end{equation}
where $\inv{P}(a, z)$ is the inverse of the regularized lower incomplete Gamma function $P(a, z)$~\cite[Eq.~6.5.1]{Abramowitz1972}.
The maximum \glspl{zoc} for these fading types are shown in Fig.~\ref{fig:sc-max-zoc-rayleigh-nakagami}. First, recall that Rayleigh fading is a special case of Nakagami-$m$ fading. It is achieved by setting $m=1$.
From Fig.~\ref{fig:sc-max-zoc-rayleigh-nakagami}, it can be seen that the maximum \gls{zoc} increases for all fading distributions with increasing $n$. However, the maximum \gls{zoc} increases slower for higher $m$. This is due to the shape of the Gamma-distribution. For increasing $m$, the probability mass in the upper tail gets smaller. The changes in the quantile function $\inv{F_{\X}}$ are therefore very small, if we look at probabilities close to one, which is the case for $1-1/n$ with increasing $n$. Therefore the changes in the maximum \gls{zoc} are also small when $m$ is high. However, it should be emphasized that the maximum \gls{zoc} still goes to infinity for $n\to\infty$.
On the other hand, the median increases with increasing $m$ and the maximum \gls{zoc} therefore also increases with $m$ for $n=2$.

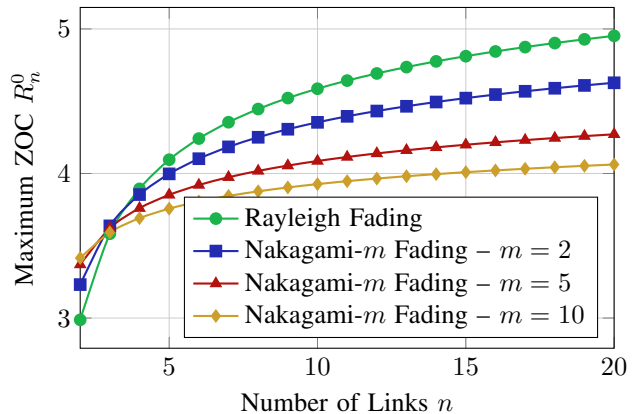
\begin{figure}
	\centering
	\begin{tikzpicture}%
	\begin{axis}[
		width=.98\linewidth,
		height=.25\textheight,
		xlabel={Number of Links $n$},
		ylabel={Maximum ZOC $\zerorate_n$},
		xmin=2, xmax=20,
		legend pos=south east,
		legend cell align=left,
		ymajorgrids,
		xmajorgrids,
		xminorgrids,
		grid style={line width=.1pt, draw=gray!20},
		major grid style={line width=.25pt,draw=gray!40},
	]
		\addplot[plot0,thick,mark=*, smooth] table[x=n,y=zocExp] {data/zoc-SC-snr10-naka2.dat};
		\addlegendentry{Rayleigh Fading};
		
		\addplot[plot1,thick,mark=square*, smooth] table[x=n,y=zocNaka] {data/zoc-SC-snr10-naka2.dat};
		\addlegendentry{Nakagami-$m$ Fading -- $m=2$};
		\addplot[plot2,thick,mark=triangle*, smooth] table[x=n,y=zocNaka] {data/zoc-SC-snr10-naka5.dat};
		\addlegendentry{Nakagami-$m$ Fading -- $m=5$};
		\addplot[plot3,thick,mark=diamond*, smooth] table[x=n,y=zocNaka] {data/zoc-SC-snr10-naka10.dat};
		\addlegendentry{Nakagami-$m$ Fading -- $m=10$};

	\end{axis}
\end{tikzpicture}
	\caption{Maximum \gls{zoc} for $n$ links with Rayleigh and Nakagami-$m$ fading with \gls{snr} $\snr=\SI{10}{\decibel}$ and \gls{sc} at the receiver.}
	\label{fig:sc-max-zoc-rayleigh-nakagami}
\end{figure}

\subsection{Summary}
In this section, we derived expressions for the maximum \gls{zoc}, if the receiver employs \gls{sc}.
An overview of the results can be found in Table~\ref{tab:sc-max-zoc-overview}.
For homogeneous marginals, we derived the maximum \gls{zoc} explicitly for all $n\geq 2$.
For $n$ heterogeneous marginals, we give an implicit characterization of the maximum \gls{zoc}.
\begin{table}[tbh]
	\renewcommand{\arraystretch}{1.25}
	\centering
	\caption{Maximum \Gls{zoc} for Selection Combining $L_{\text{SC}}(\X_1, \dots{}, \X_n) = \max (\X_1, \dots{}, \X_n)$}
	\label{tab:sc-max-zoc-overview}
	\begin{tabular}{c|c|c}
		\toprule
		\textbf{} & \textbf{Homogeneous Links} & \textbf{Heterogeneous Links}\\
		\midrule
		$\mathbf{n=2}$ & {\multirow[b]{2}{*}{$\displaystyle \log_2\left(1 +\inv{F_{\X}}\left(1-\frac{1}{n}\right)\right)$}} & $\displaystyle \log_2\left(1 + \inv{F_{\X_1}}(p^{\star})\right)$\\
		$\mathbf{n>2}$ & & $\displaystyle \sum_{i=1}^{n} F_{\X_i}\!\left(2^{\zerorate_n} - 1\right) = n-1$\\
		\bottomrule
	\end{tabular}
\end{table}
\section{Conclusion}\label{sec:conclusion}

In this work, we investigated the \gls{zoc} for fading channels with a dependency structure. Interestingly, there exist joint distributions for which the \gls{zoc} is strictly positive, without requiring perfect \gls{csit}, i.e., without power control and only with an average power constraint. First, we provide a parameterized description of a dependency structure in form of a copula that achieves positive \glspl{zoc}.
This shows that the set of joint distribution, for which the \gls{zoc} is positive, is not a singleton. Next, we investigated the maximum \gls{zoc} over all joint distributions with given marginals. 
In the homogeneous case, i.e., all marginal distributions are the same, we provide an explicit expression for the maximum \gls{zoc} when \gls{sc} is used at the receiver as diversity combining technique. For \gls{mrc} at the receiver, we derive upper and lower bounds on the maximum \gls{zoc}. The gap between these bounds is finite for marginal distributions that fulfill certain mild requirements. For heterogeneous marginals, we describe the solution in the case of $n=2$. The general $n$-dimensional case with heterogeneous marginal distributions remains an open problem at this moment.

This work gives a theoretical analysis and shows how dependency control among random variables can enable ultra-reliable communications, which is of particular interest for \gls{urllc}. In future work, it will be necessary to address ways of implementing active dependency control in real communication systems. Key technologies enabling this might be smart relaying~\cite{Wang2008} or \glspl{ris}~\cite{DiRenzo2019}.
If a future technology allows a flexible tuning of the radio environment, it might also be interesting to relax the constraint of fixed marginals.
Furthermore, we will consider different forms of available \gls{csi} at the transmitter in future work.

\appendices
\section{Proof of Corollary~\ref{cor:max-zoc-countermonotonic}\label{app:proof-cor-max-zero-out}}
It follows from \cite[Thm.~3.2]{Frank1987} that for a given rate, the lowest-possible outage probability $\varepsilon$ is achieved for a joint distribution following the copula~\cite[Eq.~(3.4)]{Frank1987}
\begin{equation}
	C_{\varepsilon}(a, b) = \begin{cases}
		\max \left[a + b - 1, \varepsilon\right] & \text{if}~(a, b) \in [\varepsilon, 1]^2\\
		M(a, b) & \text{otherwise.}
	\end{cases}%
\end{equation}
For some further details also see \cite{Embrechts2003,Besser2020part2}.
This dependency structure, therefore, describes the upper bound on the $\varepsilon$-outage capacity for dependent fading channels.
Since we are interested in the maximum \emph{zero}-outage capacity, we set $\varepsilon=0$ and get $C_{0}(a, b)=W(a,b)$, i.e., the Fr\'echet-Hoeffding lower bound, as optimal dependency structure.

\section{Proof of Theorem~\ref{thm:mrc-achievable-zero-out-two-links}\label{app:proof-thm-achievable-zoc}}
Let the joint distribution of $\X_1$ and $\X_2$ be determined by the copula
\begin{equation}\label{eq:copula-frank}
C_t(a, b) = \begin{cases}
\max\left[a + b - t, 0\right] & \text{if } (a, b)\in [0, t]^2\\
M(a, b) & \text{otherwise}
\end{cases}
\end{equation}
with $t\in[0, 1]$. The copula is also shown in Fig.~\ref{fig:copula-frank}. For $t=0$ and $t=1$, $C_t$ is equal to the Fr\'{e}chet-Hoeffding upper and lower bounds, respectively.
Observe that $C_t(a, b)=0$ if $a+b\leq t$. Since we are interested in the area $\mathcal{S}$, cf.~\eqref{eq:outage-prob-integral}, we need to transform the area from the copula space into the space of {$x_1$} and {$x_2$}.
From Sklar's theorem, we know that $F_{\X_1, \X_2}(x, y)=0$, if $F_{\X_1}(x_1)+F_{\X_2}(x_2)\leq t$. We will refer to this area as $\mathcal{B}$ as defined in \eqref{eq:def-area-b} in the following. %

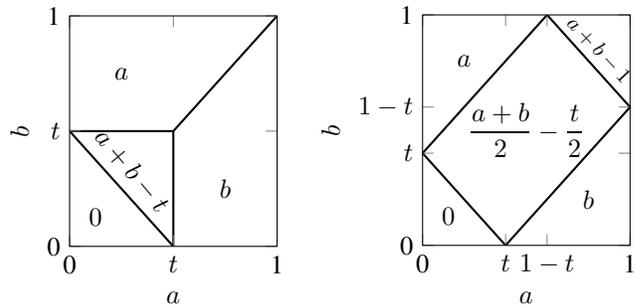
\begin{figure}[!t]
	\centering
	\subfloat[{Copula from \cite[Chap.~3.2.2]{Nelsen2006}}]{\begin{tikzpicture}
\begin{axis}[
	width=0.49\linewidth, %
	height=.19\textheight,
	xlabel={$a$},
	ylabel={$b$},
	domain=0:1,
	xmin=0,
	xmax=1,
	ymin=0,
	ymax=1,
	axis on top,
	xtick={0, 0.5, 1},
	ytick={0, 0.5, 1},
	xticklabels={0, $t$, 1},
	yticklabels={0, $t$, 1},
]

\addplot[thick, black] {-x+0.5};
\addplot[thick, black] coordinates {(0, 0.5) (0.5, 0.5)};
\addplot[thick, black] coordinates {(0.5, 0) (0.5, 0.5)};
\addplot[thick, black] coordinates {(0.5, 0.5) (1, 1)};

\node[anchor=center, text=black] at (axis cs: 0.125,0.125) {$0$};
\node[anchor=center, text=black, rotate={-50}] at (axis cs: 0.3,0.33) {$a+b-t$}; %
\node[anchor=center, text=black] at (axis cs: 0.25,0.75) {$a$};
\node[anchor=center, text=black] at (axis cs: 0.75,0.25) {$b$};

\end{axis}
\end{tikzpicture}
		\label{fig:copula-frank}
	}
	\hfil
	\subfloat[Generalized Circular Copula]{\begin{tikzpicture}
\begin{axis}[
	width=0.49\linewidth, %
	height=.19\textheight,
	xlabel={$a$},
	ylabel={$b$},
	domain=0:1,
	xmin=0,
	xmax=1,
	ymin=0,
	ymax=1,
	axis on top,
	xtick={0, 0.4, 0.6, 1},
	ytick={0, 0.4, 0.6, 1},
	xticklabels={0, $t$, $1-t$, 1},
	yticklabels={0, $t$, $1-t$, 1},
]

\addplot[thick, black] {x+0.4};
\addplot[thick, black] {x-0.4};
\addplot[thick, black] {-x+0.4};
\addplot[thick, black] {-x+1.6};
\node[anchor=center, text=black] at (axis cs: 0.125,0.125) {$0$};
\node[anchor=center, text=black, rotate={-47}, scale=.8] at (axis cs: 0.85,0.85) {$a+b-1$}; %
\node[anchor=center, text=black] at (axis cs: 0.2,0.8) {$a$};
\node[anchor=center, text=black] at (axis cs: 0.8,0.2) {$b$};
\node[anchor=center, text=black] at (axis cs: 0.5,0.5) {$\displaystyle\frac{a+b}{2}-\frac{t}{2}$};

\end{axis}
\end{tikzpicture}
		\label{fig:copula-circular}
	}
	\caption{Two different copulas with parameter $t$ which allow achieving a positive \gls{zoc}.}\label{fig:copula-examples}
\end{figure}

We denote the boundary of $\mathcal{B}$ as $B_t(x)=\inv{F_{{\X_2}}}(t-F_{{\X_1}}(x))$. Since we are looking for a tangent with slope \num{-1}, we have the following condition
\begin{equation}
	\frac{\partial B_t}{\partial {x_1}} = \frac{-f_{{\X_1}}(x)}{f_{{\X_2}}\left(\inv{F_{{\X_2}}}(t-F_{{\X_1}}(x))\right)} \stackrel{!}{=} -1 \label{eq:condition-x-star-slope-boundary}\,,
\end{equation}
where we set $u=t-F_{{\X_1}}({x_1})$ and apply the well-known rules for derivatives.
The solutions to \eqref{eq:condition-x-star-slope-boundary} are denoted as $x^{(i)}$.
Recall that we are looking for lines in the form of $x_1+x_2=s$ and that the above condition gives these tangents on $B$ at $x^{(i)}$. The next candidate solutions are therefore in the form $x_1+x_2=x^{(i)}+B_t(x^{(i)})$. An exemplary $\mathcal{B}$ can be seen in Fig.~\ref{fig:area-bounds-xy-space}. It is easy to see that there exist multiple solutions $x^{(i)}$ to \eqref{eq:condition-x-star-slope-boundary}. However, only one of the lines $x^{(i)}+B(x^{(i)})$ lies completely in $\mathcal{B}$. For the optimal solution $x^\star$, we, therefore, need to take the minimum over all $x^{(i)}$, i.e.,
$x^\star = \argmin_{x^{(i)}} \left\{x^{(i)} + B(x^{(i)})\right\}$.

If the boundary $\mathcal{B}$ is not convex or it has no slope of \num{-1} in the first quadrant of the $x_1$-$x_2$ plane, we need to consider the limit points, which are given at $x_1=0$ and $x_2=0$ and we have $x_2\leq\inv{F_{\X_2}}(t)$ and $x_1\leq\inv{F_{\X_1}}(t)$, respectively. Recall that the original problem $x_1+x_2\leq s$ can be viewed as finding the line $x_2=s-x_1$ with the maximum $s$ such that the line is in $\mathcal{B}$.
With the first bounds on $x_1$ and $x_2$, we have the simple lines $x_1+x_2=\inv{F_{\X_1}}(t)$ and $x_1+x_2=\inv{F_{\X_2}}(t)$ as possible candidates.
This is illustrated in Fig.~\ref{fig:area-bounds-xy-space}.
From this exemplary plot, it is also easy to see that both lines do not fully lie in $\mathcal{B}$.
Thus, there is a positive probability mass in the areas $x_1+x_2\leq\inv{F_{\X_1}}(t)$ and $x_1+x_2\leq\inv{F_{\X_2}}(t)$.
In this case, we need to find the tangent on the boundary of $\mathcal{B}$.
We now have the following possible candidates for $s$
\begin{equation*}
	s_1 = x^\star+B_t(x^\star),\qquad
	s_2 = \inv{F_{{\X_1}}}(t),\qquad
	s_3 = \inv{F_{{\X_2}}}(t).
\end{equation*}
Recall that $s=2^{\zerorate}-1$ with the \gls{zoc} $\zerorate$ when the joint distribution of $\X_1$ and $\X_2$ follow copula $C_t$.
With reference to Fig.~\ref{fig:area-bounds-xy-space}, it becomes clear that we need to take the minimum of all $s_i$ in order to guarantee that $x_1+x_2\leq s$ is a subset of $\mathcal{B}$. %
Combining all of the above finally states the theorem.

\begin{figure}
	\centering
	\begin{tikzpicture}
\begin{axis}[
	width=.96\linewidth,
	height=.25\textheight,
	xlabel={$x_1$},
	ylabel={$x_2$},
	domain=0:6,
	xmin=0,
	xmax=2.75,
	ymin=0,
	ymax=2.75,
	axis on top,
	xtick={0, 1.5, 2.242},
	xticklabels={0, $B_t(x^\star)+x^\star$, $\inv{F_{\X_1}}(t)$},
	ytick={0, 2},
	yticklabels={0, $\inv{F_{\X_2}}(t)$},
]

\addplot[plot4, thick, fill=plot4, fill opacity=0.2, area legend, samples=100] {.5*sin(deg(4*x))-x+2}\closedcycle;
\addlegendentry{$F_{\X_1}(x_1)+F_{\X_2}(x_2)\leq t$};
\node[anchor=south west, text=plot4] at (axis cs: 0.4,0.5) {$\mathcal{B}$};

\addplot[plot2, thick] {2-x};
\addlegendentry{$x_1+x_2 = \inv{F_{\X_2}}(t)$};

\addplot[plot3, thick] {2.242-x};
\addlegendentry{$x_1+x_2 = \inv{F_{\X_1}}(t)$};

\addplot[black, dashed, thick] {2.5-x};
\addlegendentry{$x_1+x_2 = B_t(x^{(1)})+x^{(1)}$};

\addplot[black, thick] {1.5-x};
\addlegendentry{$x_1+x_2 = B_t(x^\star)+x^\star$};

\end{axis}
\end{tikzpicture}
	\caption{Area with zero probability mass of the joint distribution $F_{\X_1, \X_2}$ and different candidates for corresponding \gls{zoc}.}
	\label{fig:area-bounds-xy-space}
\end{figure}
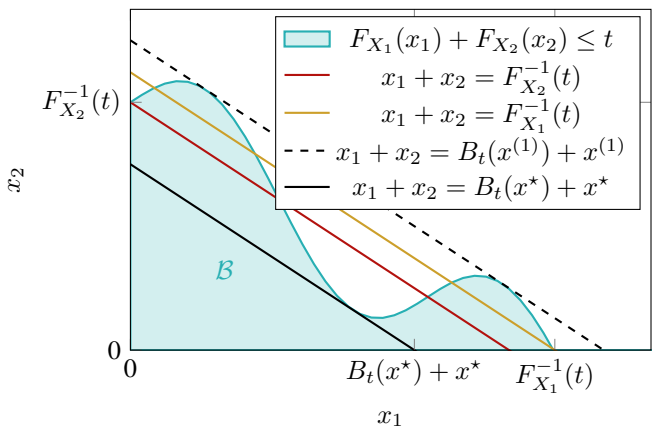

\section{Proof of Lemma~\ref{lem:dist-cond-w}\label{app:proof-lem-dist-cond-w}}
We start with the following observation.
Since we want to minimize the sum $\sum_{i=1}^{n}x_i$, we are interested in the tangent planes $\sum_{i=1}^{n}x_i=s$ on $B(\vec{x})=0$.
The condition for the critical points is therefore given by
\begin{equation}
	\frac{\partial B}{\partial x_i} = \frac{\partial B}{\partial x_j}, \quad \forall i, j\,.
\end{equation}
For the considered $B$, this means that
\begin{equation}\label{eq:cond-same-derivative}
	f(x_i) = f(x_j), \quad \forall i, j \,.
\end{equation}
First, it is easy to see from this, that there will always be a critical point on the identity line $x_1=\cdots{}=x_n$ that solves \eqref{eq:cond-same-derivative}.
If the density $f$ is strictly decreasing, e.g., the exponential distribution, it has an inverse function and the point on the identity line will be the only solution \eqref{eq:cond-same-derivative}. For such distributions, the proof is concluded.

For general distributions with quasiconcave densities, we next check the curvature of $B$ at the point on the identity line that we are interested in. If it is negative, this corresponds to a maximum of the curve and is \emph{not} a solution of \eqref{eq:opt-problem-nice-dist}. If we have a minimum on the identity line, it is a valid candidate to be the solution to \eqref{eq:opt-problem-nice-dist}.

With the implicit function theorem~\cite[Sec.~8.3]{Binmore2002}, we will express the implicit function $B(x_1, \dots{}, x_n)=0$ by an explicit function $b(x_1, \dots{}, x_{n-1})$ as
\begin{equation}
	B(x_1, \dots{}, x_n) = 0 \Leftrightarrow b(x_1, \dots{}, x_{n-1}) = x_n \,.
\end{equation}
From the implicit function theorem, we can determine the gradient of $b$ as
\begin{equation}
	\frac{\partial b}{\partial x_i} = -\inv{\left(\frac{\partial B}{\partial x_n}\right)} \frac{\partial B}{\partial x_i}, \quad i=1, \dots{}, n-1\,,
\end{equation}
and the second derivatives, i.e., the Hessian matrix, whose entries are then given by %
\begin{equation}\label{eq:second-deriv-b}
	\frac{\partial^2 b}{\partial x_i \partial x_j} = -{\left(\frac{\partial B}{\partial x_n}\right)}^{-2}\left(\frac{\partial^2 B}{\partial x_i \partial x_j} \frac{\partial B}{\partial x_n} - \frac{\partial B}{\partial x_i} \frac{\partial^2 B}{\partial x_n \partial x_j} \right)\,.
\end{equation}
For the considered $B$, this gives
\begin{equation}\label{eq:second-deriv-main-diag-b-w}
	\frac{\partial^2 b}{\partial x_i^2} = -\frac{f'(x_i)}{f(x_n)}
\end{equation}
and
\begin{equation}
	\frac{\partial^2 b}{\partial x_i \partial x_j} = 0\,, \quad i\neq j,
\end{equation}
where $f'$ is the derivative of the density.
From this, it can be seen that the Hessian matrix of $B$ is a diagonal matrix and its eigenvalues are therefore the entries on the main diagonal which are given by \eqref{eq:second-deriv-main-diag-b-w}. In order to have a minimum on the identity line, the Hessian needs to be positive definite at this point, i.e., all of its eigenvalues need to be positive.
From \eqref{eq:second-deriv-main-diag-b-w}, we see that this is the case if
\begin{equation}
	f'(x^\star) = f'\left(\inv{F}\left(1-\frac{1}{n}\right)\right) < 0\,,
\end{equation}
which is exactly \eqref{eq:cond-nice-w}.

We have shown that there is a local minimum on the identity line for distributions that fulfill the above condition, and we now show that this is then also a global minimum.
For this, recall that all other critical points are given by \eqref{eq:cond-same-derivative}.
For the general case of a strictly quasiconcave density, we know that there are two solutions to $f(x) = \alpha$
for each level set $\alpha$, due to the unimodality.
One of the solutions is greater than the mode while the other is less than the mode.
If all $x_i$ are on the same side from the mode, they are all equal and this is, therefore, the point on the identity line that we already considered.
We, therefore, have to check the remaining solutions, i.e., when at least two $x_i$ are on the different sides of the mode.
In this case, the $x_i$ that is less than the mode is on the increasing part of the density, i.e., $f'(x_i)>0$, while $f'(x_j)<0$, when $x_j$ is greater than the mode.
With \eqref{eq:second-deriv-main-diag-b-w}, we can therefore conclude that the Hessian matrix is indefinite at all remaining critical points, i.e., they are saddle points.

Since we assume unbounded support of $F_{\X}$, the range of each $x_i$ is $[0, \infty)$.
Therefore, if $x_i\to\infty$ for any $i$, the sum $x_1 + \cdots{} + x_n$ also tends to $\infty$.
In other words, the function is coercive~\cite[Def.~2.31]{Beck2014} with respect to the sum.

We can summarize the above discussion as follows. A distribution that fulfills condition \eqref{eq:cond-nice-w} from the lemma has a unique minimum of the sum on the identity line. All other stationary points are saddle points.
Combining this with the coerciveness, we can conclude that the minimum on the identity line is global.

\section{Proof of Lemma~\ref{lem:dist-cond-arch-copula}\label{app:proof-lem-dist-cond-arch-copula}}
The proof of Lemma~\ref{lem:dist-cond-arch-copula} follows the same idea as the proof of Lemma~\ref{lem:dist-cond-w}.
However, we now consider
$B(x_1,\dots{},x_n) = \sum_{i=1}^{n}\left(F_{\X}(x_i)\right)^{\frac{1}{n-1}}-n+1$.
The first derivatives of $B$ are given as
\begin{equation}
	\frac{\partial B}{\partial x_i} = g(x_i) = \frac{1}{n-1} \left(F(x_i)\right)^{\frac{2-n}{n-1}}f(x_i)\,.
\end{equation}
Thus, the critical points for the tangent plane are now given by the condition
\begin{equation}
	\left(F(x_i)\right)^{\frac{2-n}{n-1}}f(x_i) = \left(F(x_j)\right)^{\frac{2-n}{n-1}}f(x_j), \quad \forall i, j \,.
\end{equation}

The second partial derivatives of $B$ are given as
\begin{equation}\label{eq:second-deriv-main-diag-b-arch-copula}
	\frac{\partial^2 B}{\partial x_i^2} = g'(x_i) = \frac{\left(F(x_i)\right)^{\frac{2-n}{n-1}}}{n-1} \left(f'(x_i) + \frac{2-n}{n-1} \frac{\left(f(x_i)\right)^2}{F(x_i)}\right)
\end{equation}
and
\begin{equation}
	\frac{\partial^2 B}{\partial x_i \partial x_j} = 0, \quad i\neq j \,.
\end{equation}
Combining this with \eqref{eq:second-deriv-b}, we can see that the Hessian matrix again is a diagonal matrix, where the entries on the main diagonal are given by \eqref{eq:second-deriv-main-diag-b-arch-copula}.
Similarly to the proof of Lemma~\ref{lem:dist-cond-w}, we have a minimum on the identity line, if the Hessian matrix is positive definite at this point.
From the condition $B(x_1,\dots{},x_n)=0$, we get the point on the identity line as
\begin{equation*}
	x_1 = \dots{} = x_n = x^\star = \inv{F}\left(\left(1-\frac{1}{n}\right)^{n-1}\right)\,.
\end{equation*}
All of the eigenvalues of the Hessian are positive if
\begin{equation}
	f'(x^\star) < -\frac{2-n}{n-1} \left(1-\frac{1}{n}\right)^{1-n} f(x^\star)^2\,.
\end{equation}
This is condition \eqref{eq:cond-nice-arch-copula} from the lemma to prove.

By the assumption in the lemma, $g(x)$ is a unimodal function.
We can therefore use the same argument as in the proof of Lemma~\ref{lem:dist-cond-w} in App.~\ref{app:proof-lem-dist-cond-w}.
With the above observations and reference to App.~\ref{app:proof-lem-dist-cond-w}, this concludes the proof.

\section{Proof of Theorem~\ref{thm:mrc-outer-bound-zoc-n-links-w}\label{app:proof-thm-mrc-outer-w}}
The general idea of the proof is similar to the one of Theorem~\ref{thm:mrc-achievable-zero-out-two-links} for two links.
However, we need to do the following adjustments.
Recall that we want to find the hyperplane of the form
\begin{equation}\label{eq:def-hyperplane-n-links}
	\sum_{i=1}^{n} x_i = s,
\end{equation}
which touches $\mathcal{B}$ given by
$\sum_{i=1}^{n}F_{\X_i}(x_i) + 1 - n = 0$.
Since we assume that the distribution is $B$-\nice{}, we know this tangent point $x^\star$ will be on the identity line, i.e., $x^\star = x_1 = \dots{} = x_n$.
We therefore get
$x^\star = \inv{F_{\X}}\left(\frac{n-1}{n}\right)$,
and applying this to \eqref{eq:def-hyperplane-n-links} gives the corresponding $s^\star$ as
$s^\star = n x^\star = n\inv{F_{\X}}\left(\frac{n-1}{n}\right)$.
With the definition of $s$, we get the statement \eqref{eq:mrc-outer-bound-zoc-n-links-w} from the theorem.

\section{Proof of Theorem~\ref{thm:mrc-inner-bound-zoc-n-links}\label{app:proof-thm-mrc-inner}}
The proof is basically the same as the one of Theorem~\ref{thm:mrc-outer-bound-zoc-n-links-w}. The only difference is that we use the lower bound on the Archimedean copulas~\cite[Prop.~4.6]{McNeil2009} (also see \cite[Prop.~2]{Lee2014})
\begin{equation*}
	C_n(u_1,\dots{},u_n) = \left(\positive{\sum_{i=1}^{n} u_i^{\frac{1}{n-1}}-n+1}\right)^{n-1}\,.
\end{equation*}
Since this is a valid copula, the derived bound in \eqref{eq:mrc-inner-bound-zoc-n-links} is achievable and therefore an inner bound on the maximum \gls{zoc}.

\section{Proof of Corollary~\ref{cor:mrc-max-gap-inner-outer-bounds-n-links}\label{app:proof-thm-mrc-max-gap}}
The difference of $\overline{\zerorate_n}$ in \eqref{eq:mrc-outer-bound-zoc-n-links-joint-mix} and $\underline{\zerorate_n}$ in \eqref{eq:mrc-inner-bound-zoc-n-links} can be written as
\begin{equation*}
	\overline{\zerorate_n} - \underline{\zerorate_n} = \log_2\left(\frac{1 + n\expect{\X}}{1 + n\inv{F_{\X}}\left(\left(1-\frac{1}{n}\right)^{n-1}\right)}\right)\,.
\end{equation*}

The gap increases with increasing $n$ and is therefore upper bounded by
\begin{equation*}
	\lim_{n\to\infty} \overline{\zerorate_n} - \underline{\zerorate_n} = \log_2\left(\frac{\expect{\X}}{\inv{F_{\X}}(\frac{1}{\e})}\right)\,,
\end{equation*}
which is \eqref{eq:mrc-max-gap-inner-outer-bounds-n-links}.

\printbibliography

\end{document}